\documentclass[11pt,a4paper]{article}
\pdfoutput=1

\usepackage[margin=2.6cm]{geometry}
\usepackage{amsmath}
\usepackage{amssymb}
\usepackage{amsthm}
\usepackage{array}
\usepackage[hidelinks]{hyperref}
\hypersetup{pdftitle={Bath dimension and initial entropy for closed repeated use of a quantum channel},
  pdfauthor={Seth Douglas}}

\numberwithin{equation}{section}
\theoremstyle{plain}
\newtheorem{theorem}{Theorem}[section]
\newtheorem{proposition}[theorem]{Proposition}
\newtheorem{lemma}[theorem]{Lemma}
\newtheorem{corollary}[theorem]{Corollary}
\theoremstyle{definition}
\newtheorem{definition}[theorem]{Definition}

\DeclareMathOperator{\Tr}{Tr}
\DeclareMathOperator{\rank}{rank}
\DeclareMathOperator{\supp}{supp}

\newcommand{\rg}[1]{\mathsf{#1}}

\newcommand{\id}{\mathrm{id}}
\newcommand{\one}{\mathbb{I}}
\newcommand{\CC}{\mathbb{C}}
\newcommand{\Dist}{D}                      
\newcommand{\Ddiam}{D_{\diamond}}          
\newcommand{\Rmap}{\mathcal{R}}            
\newcommand{\Pfun}{\mathsf{P}_{\mathrm{q}}}
\newcommand{\kexact}{\kappa_{\mathrm{exact}}}
\newcommand{\cseq}{c_{\mathrm{seq}}}
\newcommand{\ket}[1]{\lvert #1 \rangle}
\newcommand{\bra}[1]{\langle #1 \rvert}
\newcommand{\ketbra}[2]{\lvert #1 \rangle\langle #2 \rvert}

\title{Bath dimension and initial entropy for closed repeated use of a
quantum channel}
\author{Seth Douglas\\\href{mailto:seth.douglas@gmail.com}{\texttt{seth.douglas@gmail.com}}}
\date{}

\begin{document}
\maketitle

\begin{abstract}
We characterize the bath resources needed to supply repeated uses of a fixed
finite-dimensional quantum channel in a closed device. For each horizon $T$,
one bath, one initial state and one repeated unitary are fixed before the user.
Each output is returned before the next input arrives; no reset, discard,
fresh ancilla or uncounted controller is available. Approximation error must
vanish against arbitrary adaptive users with quantum memory and references.
Writing $r=\lim \log_2(R_T)/T$ for the bath dimension rate and
$s=\lim S(\omega_T)/T$ for the actual initial entropy rate, we prove that the
achievable region is exactly $s\ge 0$, $r+s\ge h$ and $r-s\ge\kappa$. Here $h$ is
maximum entropy exchange and $\kappa$ is a smoothed independent-reference
extension cost, with the zero-error limit taken before the supremum over
full-rank inputs. Its exact fixed-input form is an affine transform of the
zero-leakage quantum privacy funnel. The minimum dimension rate is
$(h+\kappa)/2$. The proof combines entropy converses, a
bath-dimension-independent support repair, and a closed adaptive implementation
of encoder-only fully quantum Slepian--Wolf recycling. All seeds, clocks,
workspace and retained residues are counted. Worked examples include
dephasing, pure replacement and a qubit channel with $0<\kappa<h$. No
computability of $\kappa$ or efficient circuit synthesis is claimed.
\end{abstract}

\section{Introduction and model}
\label{sec:model}

A quantum channel is normally charged once. Fix a CPTP map $\Phi$ and a single
use costs a dilation: an environment of dimension at most $q^2$ if it starts
pure, or some other finite bath if it may start mixed. Serving $T$ inputs in
sequence is a different question, and the difference is not merely a factor
of $T$: the closed cost is still linear in $T$, but its coefficient, and the
trade between dimension and initial entropy, are what a single-use dilation
does not determine. A
device that has already produced outputs must have somewhere to keep what
producing them left behind. If it may reset its bath between visits, discard a
spent cell, or draw a fresh ancilla, the question collapses to the single-use
one: use one dilation $T$ times and let the rest arrive from outside. Those
three permissions are exactly what a closed implementation lacks. Reset
consumes a supply of pure states, discard needs a place to put what is
discarded, and a fresh ancilla is a resource that somebody prepared. Charging
all of them means fixing one finite Hilbert space and one finite initial state
in advance and asking what fits inside.

Definition~\ref{def:closed-device} is that contract. For each horizon $T$ the
device is one active bath $\rg{B}_T$ of dimension $R_T$, one user-independent
initial state $\omega_T$ on it, and one unitary $W_T$ applied at every visit;
all three are chosen before the user and may depend on $\Phi$ and $T$. The
output is handed back before the next input arrives. The tester is adaptive,
may retain every output, and may hold arbitrary finite references, so the
error of Definition~\ref{def:complete-error-rates} compares complete final
states, not single-use marginals. Nothing inside is free: the seed, the clock,
the workspace, the records and every register parked permanently after use are
all part of $R_T$. The one object outside is a purifier of $\omega_T$, which
the device may never touch and which is therefore never a resource.

Two numbers describe such a device, and they are not interchangeable. The
first is $\log R_T$, the dimension it occupies. The second is $S(\omega_T)$,
the actual entropy of the state it was initialized in. A rank count lies
between them: $S(\omega)\le\log\rank(\omega)\le\log R$, with equality
throughout only for a flat positive spectrum, so neither dimension nor rank
determines the other number. The four-dimensional example after
Definition~\ref{def:collision} separates all three. Initial entropy is the
right second axis because the contract makes it a resource that can be spent:
a device may be handed a mixed initializer, and what it is charged is the
entropy of that initializer, not the provenance of the randomness, which lies
outside this contract and is not identified here with work or energy. The
region below then charges that convenience back. Past the corner, every
further bit of initial entropy costs a bit of dimension, since $r\ge\kappa+s$
there.

Theorem~\ref{thm:rate-region} characterizes what is achievable. Writing $r$
for the dimension rate and $s$ for the actual initial entropy rate, the closed
region is $s\ge 0$, $r+s\ge h(\Phi)$ and $r-s\ge\kappa(\Phi)$, and the least
dimension rate is $(h+\kappa)/2$, attained with initial entropy rate
$(h-\kappa)/2$ and a spectrum flat on its support. Both channel quantities are
intrinsic. The first, $h$, is Schumacher's entropy exchange maximized over
inputs. The second, $\kappa$, is an infimum of $S(\rg{Q}\rg{A}\mid\rg{Z})$
over finite extensions whose reference stays exactly in product with the
auxiliary, smoothed in the output error, with the positive-error limit taken
before the full-rank input supremum; Corollary~\ref{cor:pointwise-funnel}
identifies the positive-error limit at each fixed full-rank input with
$S(\chi_\rho)-\Pfun^{\psi_{\rg{E}\rg{Q}}}(0)$, an affine transform of the
zero-leakage quantum privacy funnel. Kappa is the supremum of these limits.
Two resources, two channel quantities, one closed region.

The closest antecedent contract is Ryb\'ar and Ziman's repeatable quantum
memory channel: one fixed finite memory, initializer and unitary, required to
produce the same single-use marginal at every repetition, for uncorrelated
inputs. Ryb\'ar and Ziman also discuss finite $n$-repeatability and a finite-cell
construction. Here we optimize dimension and actual initialization entropy
as the horizon grows, while requiring vanishing complete error against
arbitrary adaptive quantum users. The gap between the contracts is not cosmetic. One fair shared bath
bit controlling $\one$ or $Z$ is a legitimate marginal dephasing repeater in
their sense, and Proposition~\ref{prop:shared-seed} shows its complete service
error is at least $1-2^{1-T}$, already $1/2$ at $T=2$. No resolution of their
original question is claimed here.

The contributions, in decreasing operational significance, are these. First,
the characterization itself under the closed contract, together with its
converse Theorem~\ref{thm:same-spectrum-converse}, which assumes no
entropy-rate limit, constrains the same actual initialization through two
separate legal experiments, and holds at every finite horizon with constants
explicit in that horizon's own error. Second, the active support
repair of Theorem~\ref{thm:active-repair}: an approximate finite collision is
corrected so that its Choi support lies inside the target's, with the positive
initial spectrum unchanged and a gain penalty depending only on the fixed
channel, not on the bath dimension, the initial spectrum or any auxiliary
dimension; the subsequent flagged exactification changes the spectrum
explicitly and says so, and with Theorem~\ref{thm:uniform-gain} it gives
$\kexact=\kappa$. Third, the implementation of
Theorem~\ref{thm:causal-balancing} and Appendix~\ref{app:balancing}, which
turns the encoder-only fully quantum Slepian--Wolf split into a device obeying
Definition~\ref{def:closed-device}: spectral concentration uniform over
adaptive controllers, one physical seed reused across every block and returned
jointly independent of the whole future-accessed exterior, reversible
flat-stock preparation on exact integer wires, and every residue allocated and
counted in one closed bath. The entropy-exchange functional, the privacy
funnel, the classical minimax step, the Slepian--Wolf split with its half-sum
rates and the decoupling second moment are established ingredients, cited at
their uses; the contract they are assembled under, and the repair and
allocation needed to meet it, are the implementation tasks addressed here.

The characterization has real limits, stated where they arise. The quantity
$\kappa$ is an ordered limit of an infimum with no attainment, no bound on the
auxiliary dimension, and no assertion of computability or of continuity in
$\Phi$; the thresholds in Section~\ref{sec:region} are finite but not
effective. Proposition~\ref{prop:kappa-lower} supplies a lower bound on
$\kappa$ that is an optimization over the input space alone, and
Proposition~\ref{prop:interior} brackets a single-qubit channel with
$0<\kappa<h$, so the non-degenerate case $0<\kappa<h$ occurs and the two
degenerate edges are not the only cases. Nothing here bounds circuit size or
latency, requires the bath to be returned, or extends beyond one fixed
memoryless channel visited once per input.

Section~\ref{sec:main-converse} states the region and proves the converse.
Section~\ref{sec:collision-gain} identifies $\kappa$ with an optimum over
physical approximate collisions, Section~\ref{sec:repair} repairs such a
collision into an exact one and derives the privacy-funnel form, and
Section~\ref{sec:balancing} with Appendix~\ref{app:balancing} builds an
all-horizon device from one exact witness. Section~\ref{sec:region} assembles
the region from these, Section~\ref{sec:examples} gives worked channels and
the limits above, and Section~\ref{sec:related} compares the contract with
neighbouring ones.

\begin{definition}[closed causal device]
\label{def:closed-device}
Fix a CPTP map $\Phi$ from operators on $\rg{S}=\CC^q$ to operators on
$\rg{A}=\CC^q$. For every integer horizon $T\ge 1$, a device consists of a
finite active bath $\rg{B}_T$ of dimension $R_T$, a density operator
$\omega_T$ on it, and a unitary $W_T:\rg{S}\rg{B}_T\to\rg{A}\rg{B}_T$, with
$\rg{S}$ and $\rg{A}$ identified between visits. All three may depend on
$\Phi$ and $T$ and are chosen before the user. Initially $\omega_T$ is in
product with the user's entire state. A mathematical purifier $\rg{F}$ is
inaccessible forever; it is an analysis device only, the device never acts on
it, so it is not charged, and the cost of a mixed $\omega_T$ enters through
$S(\omega_T)$. Each visit applies this same $W_T$ and hands $\rg{A}$
back before the next input arrives. All device seeds, clocks, stock, work
wires, records and permanently parked residues are part of $\rg{B}_T$.
Offline preparation and unrestricted fixed hardware are allowed. No device
reset, discard, fresh online ancilla, external randomness, purifier operation
or uncounted controller is allowed. There is no bath-return requirement,
marginal or independent.
\end{definition}

Here \emph{active bath} means the entire allocated device Hilbert space
$\rg{B}_T$, including registers parked permanently after use; the
\emph{active dimension} is $R_T$. The inaccessible purifier is excluded. A
close antecedent is the unitary memory model of Ryb\'ar and Ziman
\cite[sections II--III]{RybarZiman2008}. Their infinite-repeatability definition
fixes one finite device for exact single-use marginals with uncorrelated inputs;
they also discuss finite $n$-repeatability and preallocated cells. Here we optimize
horizon-dependent resources while approximating the complete process against
adaptive inputs. Section~\ref{sec:related} gives an explicit separation; no
resolution of their original question is claimed.

\begin{definition}[complete error and rates]
\label{def:complete-error-rates}
A tester has arbitrary finite user memory and references, may retain all
outputs, and applies arbitrary adaptive quantum operations between visits. Let
$\rho_T^{\mathrm{real}}(U)$ and $\rho_T^{\mathrm{ideal}}(U)$ be its complete
final states when connected to the device or to $T$ independent $\Phi$ calls.
With $\Dist(\rho,\sigma)=\lVert\rho-\sigma\rVert_1/2$, set
\begin{equation}
\label{eq:complete-error}
\delta_T=\sup_U \Dist\bigl(\rho_T^{\mathrm{real}}(U),
\rho_T^{\mathrm{ideal}}(U)\bigr),\qquad
\cseq(\Phi)=\inf_{\delta_T\to 0}\ \limsup_{T\to\infty}\frac{\log R_T}{T}.
\end{equation}
The supremum includes all finite reference sizes. \emph{All-horizon} means a
device for every positive integer $T$, with error tending to zero along all
integers. An achievable finite rate pair has actual limits $\log R_T/T\to r$
and $S(\omega_T)/T\to s$; the rate region is the closure of these pairs. Thus
$R_T=2^{rT+o(T)}$: positive $r$ means exponential growth, while $r=0$ permits
subexponential growth.
\end{definition}

Logarithms and entropies are base two, $S(\rho)=-\Tr\rho\log\rho$, with
$0\log 0=0$. An unconditional stopped experiment is covered by padding its
continuation with dummy inputs. No bound conditional on a rare stopping
outcome is promised. Batch access to all inputs would remove the
immediate-output requirement and is a different task.

\begin{definition}[intrinsic costs]
\label{def:intrinsic-costs}
For a minimal pure Stinespring dilation $V:\rg{S}\to\rg{A}\rg{E}$ of $\Phi$,
write $\Phi^c(\rho)=\Tr_{\rg{A}}V\rho V^*$ and
$h(\Phi)=\max_\rho S(\Phi^c(\rho))$. Let
$\ket{\Omega_q}=q^{-1/2}\sum_i\ket{i,i}$ and
$J_\Phi=(\id\otimes\Phi)(\ketbra{\Omega_q}{\Omega_q})$, so that
$\Tr J_\Phi=1$ and $\Tr_{\rg{A}}J_\Phi=\one_{\rg{Q}}/q$. For $\rho>0$ choose
$\ket{\psi_\rho}=(\one\otimes\sqrt{q\rho})\ket{\Omega_q}$, whose $\rg{Q}$
marginal is $\rho^{T}$ in this basis, and put
$\chi_\rho=(\id\otimes\Phi)(\ketbra{\psi_\rho}{\psi_\rho})$. For
$\varepsilon>0$ define
\begin{equation}
\label{eq:kappa-definition}
k_{\rho,\varepsilon}=\inf\ S(\rg{Q}\rg{A}\mid\rg{Z})_\sigma
\quad\text{subject to}\quad
\sigma_{\rg{Q}\rg{Z}}=\rho^{T}\otimes\sigma_{\rg{Z}}\ \text{exactly},\quad
\Dist(\sigma_{\rg{Q}\rg{A}},\chi_\rho)\le\varepsilon,
\end{equation}
and $\kappa(\Phi)=\sup_{\rho>0}\lim_{\varepsilon\downarrow 0}
k_{\rho,\varepsilon}$.
\end{definition}

The quantity $h$ maximizes the established entropy exchange of Schumacher
\cite[section V\,A, summary (iii), p.~2622]{Schumacher1996}. Complementary
dilations differ by environment isometries, so this number is independent of
the minimal choice. In \eqref{eq:kappa-definition} the infimum runs over
normalized states on $\rg{Q}\rg{A}\rg{Z}$ and every finite $\rg{Z}$; there is
no uniform dimension bound or attainment assertion. Here
$S(X\mid Y)=S(XY)-S(Y)$. The value $k_{\rho,\varepsilon}$ is nonincreasing in
$\varepsilon$ and lies between $S(\rho)-\log q$ and $S(\chi_\rho)$, so its
positive-error limit exists and is finite. The positive-error limit is taken
with $\rho$ fixed, before the full-rank supremum. Equality with the
exact-extension value at $\varepsilon=0$ requires the stability argument after
Theorem~\ref{thm:active-repair}; it is not an assumption of this definition.

The exact fixed-input optimization is already related to a named entropy
functional. Write $K_0(\rho)$ for the same infimum with
$\sigma_{\rg{Q}\rg{A}}=\chi_\rho$ exactly. For
$\psi_{\rg{Q}\rg{A}\rg{E}}=(\id\otimes V)\psi_\rho$, all finite
$\rg{Q}\rg{A}$ extensions are obtained by a channel $\rg{E}\to\rg{Z}$, and the
product constraint is $I(\rg{Q};\rg{Z})=0$. Thus
\begin{equation}
\label{eq:exact-funnel}
K_0(\rho)=S(\chi_\rho)-\Pfun^{\psi_{\rg{E}\rg{Q}}}(0),
\end{equation}
where $\Pfun$ is the quantum privacy funnel of Datta, Hirche and Winter
\cite[section V, Eq.~(22)]{Datta2019}, with their $X=\rg{E}$, $Y=\rg{Q}$,
$R=\rg{A}$ and $W=\rg{Z}$. Its objective $I(YR;W)$ becomes
$I(\rg{Q}\rg{A};\rg{Z})$, and
$S(\rg{Q}\rg{A}\mid\rg{Z})=S(\chi_\rho)-I(\rg{Q}\rg{A};\rg{Z})$. Their
subscript $\mathrm{q}$ denotes quantum, not the dimension here. Supremum and
infimum range over all finite auxiliary dimensions; neither attainment nor a
dimension bound is asserted. Section~\ref{sec:repair} supplies the additional
positive-error stability for this paper's fixed channel. The known functional
and the Slepian--Wolf split credited in Section~\ref{sec:balancing} are
ingredients in the closed adaptive rate characterization proved below.

\begin{definition}[finite collision]
\label{def:collision}
A finite unitary $U$ on $\rg{S}\rg{B}$ and a state $\tau_{\rg{B}}$ give
$\Phi_C(\rho)=\Tr_{\rg{B}}U(\rho\otimes\tau)U^*$ and the \emph{active} bath
channel $\Gamma_C(\rho)=\Tr_{\rg{A}}U(\rho\otimes\tau)U^*$. Define
$f_C(\rho)=S(\Gamma_C(\rho))-S(\tau)$, $g_C=\max_\rho f_C(\rho)$ and
$\kexact=\inf_{\Phi_C=\Phi}g_C$. Here $\Gamma_C$ excludes the inert purifier.
\end{definition}

Channel distance is
$\Ddiam(\Psi,\Phi)=\lVert\Psi-\Phi\rVert_{\diamond}/2$, equivalently the
supremum half trace distance on reference-entangled inputs. The exact finite
collision and its spectrum are fixed before invoking a block limit.
Mixed-environment realizations and their unitary column constraints were
studied by Terhal et al.\ \cite[Eqs.~(6)--(8)]{Terhal1998}. Repeated system--bath
interactions also appear in Scarani et al.\ \cite[Eqs.~(1)--(2)]{Scarani2002}: their
thermalization model supplies identically prepared bath qubits that each
interact once with the same system. Our finite collision is one such unitary
building block; its repeated-service implementation must allocate and count
every cell within Definition~\ref{def:closed-device}'s horizon-dependent
closed bath.

For example, a four-dimensional bath initialized with eigenvalues
$(3/4,1/4,0,0)$ has log dimension $2$, log rank $1$, actual entropy
$H_2(1/4)\approx 0.8113$, and entropy deficit $2-H_2(1/4)\approx 1.1887$. Its
minimal purifier dimension is $2$, but those inert degrees of freedom are not
available workspace. In general only $S(\omega)\le\log\rank(\omega)\le\log R$
holds; equality of entropy and log rank requires a flat positive spectrum.

Conventions, register tables and the bibliography are collected in
Appendix~\ref{app:conventions}. The proofs below use no theorem about
finite-tracial closure and none about processes with more than one visit per
input.

\section{Main theorem and converse}
\label{sec:main-converse}

\begin{theorem}[rate region]
\label{thm:rate-region}
By Theorems~\ref{thm:same-spectrum-converse}, \ref{thm:uniform-gain},
\ref{thm:active-repair} and \ref{thm:causal-balancing}, proved below, the rate
region of Definition~\ref{def:complete-error-rates} is exactly
\[
s\ge 0,\qquad r+s\ge h,\qquad r-s\ge\kappa .
\]
The minimum dimension rate is $(h+\kappa)/2$, attained with actual entropy
$(h-\kappa)/2$ and a spectrum that is flat on its support. For $q=1$ all
intrinsic costs and the minimum rate vanish.
\end{theorem}

\begin{theorem}[same-spectrum converse]
\label{thm:same-spectrum-converse}
Every vanishing-error all-horizon family satisfies
\[
\liminf_{T\to\infty}\frac{\log R_T-S(\omega_T)}{T}\ge\kappa,
\qquad
\liminf_{T\to\infty}\frac{\log R_T+S(\omega_T)}{T}\ge h .
\]
No entropy-rate limit is assumed.
\end{theorem}

\begin{proof}
The entropy-decrease budget has the antecedent of Ryb\'ar and Ziman
\cite[section III, Eq.~(3.7)]{RybarZiman2008}. The reference-conditioned
telescope below strengthens its increment to the intrinsic extension cost.
Fix a horizon simulator and its actual user-independent $\omega_T$, with inert
purifier $\rg{F}$. Feed fresh purified copies of one full-rank $\rho$ and
retain $\rg{Y}_t=\rg{Q}_t\rg{A}_t$. Before each visit
$\rg{B}_{t-1}\rg{F}\rg{Y}_{<t}$ is pure. The register $\rg{Q}_t$ is
independent of $\rg{F}\rg{Y}_{<t}$, and tracing the active operation preserves
this product exactly. Global purity and the service error $\delta_T$ give
\begin{equation}
\label{eq:converse-budgets}
S(\rg{B}_t)-S(\rg{B}_{t-1})=S(\rg{Y}_t\mid \rg{F}\rg{Y}_{<t})
\ge k_{\rho,\delta_T},
\qquad
\log R_T-S(\omega_T)\ \ge\ T\,k_{\rho,\delta_T}.
\end{equation}
Explicitly, $\rg{Z}_t=\rg{F}\rg{Y}_{<t}$ is finite and
$\rg{Q}_t\rg{Z}_t$ is exactly product, since the visit acts only on its
complement $\rg{S}_t\rg{B}_{t-1}$. The $\rg{Q}\rg{A}$ marginal is within
$\delta_T$ of $\chi_\rho$ by the complete service guarantee and partial trace.
Thus $\rg{Y}_t\rg{Z}_t$ is a feasible extension even though $\dim\rg{Z}_t$
grows with $T$. The telescope starts at $S(\rg{B}_0)=S(\omega_T)$ and ends at
$S(\rg{B}_T)\le\log R_T$.

Two points about that extension each deserve a sentence. First, the service
guarantee bounds the complete final state at horizon $T$, whereas the
telescope reads the state just after visit $t$. The two agree here because
this tester never touches $\rg{Y}_{\le t}$ again and every later visit acts on
$\rg{S}_{t'}\rg{B}_{t'-1}$ alone, so the $\rg{Q}_t\rg{A}_t$ marginal is frozen
from visit $t$ onward and inherits the horizon-$T$ bound. Second,
$\rg{Z}_t$ contains the purifier $\rg{F}$, which
Definition~\ref{def:closed-device} makes permanently inaccessible. That is not
a use of a forbidden resource. The auxiliary $\rg{Z}$ of
Definition~\ref{def:intrinsic-costs} is a mathematical extension of a state,
not a register that anybody operates, and $k_{\rho,\varepsilon}$ is an infimum
over all such states, so the bound needs only that one finite state on
$\rg{Q}_t\rg{A}_t\rg{Z}_t$ with the two stated properties exists. No step here
prepares, reads, measures or acts on $\rg{F}$; a device or tester that did so
would violate Definition~\ref{def:closed-device}.

For each fixed $\rho$, take the positive-error limit and then the full-rank
input supremum. Even if some $\delta_T$ vanish exactly, eventually
$\delta_T\le\varepsilon$ for each $\varepsilon>0$, which supplies the required
lower bound without identifying $k_{\rho,0}$.

In a separate legal run, feed purified copies of an input attaining $h$. The
ideal user entropy is $Th$, while purity and subadditivity give
$S(\rg{Y}^T)=S(\rg{B}_T\rg{F})\le\log R_T+S(\omega_T)$. Entropy continuity on
the $q^{2T}$-dimensional user system gives
\begin{equation}
\label{eq:h-budget}
\log R_T+S(\omega_T)\ \ge\ Th-2\delta_T T\log q-H_2(\delta_T).
\end{equation}
Both testers constrain the same numerical initialization entropy because the
device is fixed before either tester is chosen. Their final bath states can
differ. For each $\zeta>0$ the two eventual inequalities hold simultaneously,
so adding them gives $\liminf \log R_T/T\ge(h+\kappa)/2$. This proves the
converse part of Theorem~\ref{thm:rate-region}, even if $S(\omega_T)/T$
oscillates. The entropy-continuity bound used here is Audenaert's
\cite[Theorem 1]{Audenaert2007}, in its eventual small-distance regime. The
remaining implication is proved in Section~\ref{sec:region}.
\end{proof}

\section{Intrinsic gain and physical witnesses}
\label{sec:collision-gain}

\begin{theorem}[uniform collision gain]
\label{thm:uniform-gain}
One has
\[
\kappa=\lim_{\eta\downarrow 0}\ \inf_{\Ddiam(\Phi_C,\Phi)\le\eta} g_C .
\]
\end{theorem}

\begin{proof}
Let $\mathcal{C}_\eta$ be the finite collisions within half-diamond distance
$\eta$ of $\Phi$, and define
$a_\eta(\rho)=\inf_{C\in\mathcal{C}_\eta}f_C(\rho)$ and
$b_\eta=\inf_{C\in\mathcal{C}_\eta}\max_\rho f_C(\rho)$. These sets are
nonempty: a pure minimal Stinespring isometry extends to a finite square
collision.

Fix full-rank $\rho$ and a feasible $\sigma_{\rg{Q}\rg{A}\rg{Z}}$. Purify it
on $\rg{Q}\rg{A}\rg{Z}\rg{C}$. Its product $\rg{Q}\rg{Z}$ marginal has rank
$q\rank(\sigma_{\rg{Z}})$, so $\dim\rg{C}\ge\rank(\sigma_{\rg{Z}})$.
Purification uniqueness gives an isometry
$\rg{S}\rg{B}_0\to\rg{A}\rg{C}$ from
$\psi_\rho\otimes\Omega_{\rg{B}_0\rg{Z}}$, where
$\dim\rg{B}_0=\rank(\sigma_{\rg{Z}})$. Embed $\rg{B}_0$ into a bath of
dimension $\dim\rg{C}$ and complete this isometry to a square active-only
unitary. Initial positive eigenvalues are those of $\sigma_{\rg{Z}}$; the
added dimensions have zero weight. The register $\rg{Z}$ remains inert. Purity
gives $f_C(\rho)=S(\rg{Q}\rg{A}\mid\rg{Z})_\sigma$. In matching Schmidt
coordinates, any normalized pure reference-input vector has the form
$(X\otimes\one)\psi_\rho$ with $\lVert X\rVert^2\le 1/\lambda_{\min}(\rho)$.
Trace-norm contraction under this congruence gives
$\Ddiam(\Phi_C,\Phi)\le\varepsilon/\lambda_{\min}(\rho)$. Conversely a
collision in $\mathcal{C}_\eta$ yields a feasible extension at error $\eta$.
Therefore
\begin{equation}
\label{eq:lifting-sandwich}
a_{\varepsilon/\lambda_{\min}(\rho)}(\rho)\ \le\ k_{\rho,\varepsilon}\ \le\
a_\varepsilon(\rho).
\end{equation}
Their positive-error limits coincide for each \emph{fixed} $\rho$. No
eigenvalue bound uniform in $\rho$ is used. The $\rg{Q}$ marginal is $\rho^{T}$
in a fixed canonical basis; its eigenvalues equal those of the physical input.

Finite convex combinations of gain \emph{functions} are physical:
$\rg{B}=\bigoplus_j\rg{B}_j$, $\tau=\bigoplus_j p_j\tau_j$ and
$U=\bigoplus_j U_j$ give $f_C(\rho)=\sum_j p_j f_{C_j}(\rho)$ for all $\rho$,
with channel error at most $\eta$ if all branches lie in $\mathcal{C}_\eta$.
Actual initial entropy is $H(p)+\sum_j p_j S(\tau_j)$, and active dimension is
$\sum_j\dim\rg{B}_j$. The term $H(p)$ cancels in the gain because the same
sector probabilities remain at the output. It continues to count in storage.
The maximum of the mixed gain can be smaller than the average separate maxima.

For fixed $\eta$, every $f_C$ is continuous and concave on the compact
$q$-state space. Put $a=\sup_\rho a_\eta(\rho)$. Epsilon-optimal pointwise
choices and a finite open subcover supply $f_1,\dots,f_N$ with
$\min_j f_j(\rho)<a+\epsilon$ for every $\rho$. The downward set
$\mathcal{D}=\{x:\ \text{some }\rho\text{ has }x_j\le f_j(\rho)\ \text{for all }j\}$
is closed by subsequence compactness of the input states and convex by
concavity. The point $(a+\epsilon,\dots,a+\epsilon)$ is outside $\mathcal{D}$
by the finite cover. Strict finite-dimensional separation gives a normal with
nonnegative coordinates: a negative coordinate would be unbounded above on a
downward ray of $\mathcal{D}$. Normalize that nonzero normal to weights
summing to one, giving $\max_\rho\sum_j p_j f_j(\rho)<a+\epsilon$. The
physical flagged collision realizes this function, so weak minimax supplies
$b_\eta=\sup_\rho a_\eta(\rho)$. There is no infinite mixture, bath
compactness, attainment assumption or finite size bound.

Finally $S(\rho)-\log q\le f_C(\rho)\le S(\rho)+\log q$ by entropy
inequalities on $\rg{A}\rg{B}$. These bounds are independent of the bath. The
infimum $a_\eta$ is concave, so $\rho_t=(1-t)\rho+t\,\one/q$ satisfies
$a_\eta(\rho_t)\ge a_\eta(\rho)-3t\log q$. Thus at each $\eta$ the full-rank
supremum equals the all-input supremum, without assuming boundary continuity.
Shrinking error sets increase their infima, and hence
\begin{equation}
\label{eq:gain-minimax}
\lim_{\eta\downarrow 0}b_\eta
=\sup_{\eta>0}\ \sup_{\rho>0}a_\eta(\rho)
=\sup_{\rho>0}\ \sup_{\eta>0}a_\eta(\rho)=\kappa .
\end{equation}
Only two suprema have been interchanged.
\end{proof}

This yields finite approximate collisions uniformly over inputs at every
positive accuracy and gain slack. It yields no exact witness at zero error and
no uniform bound on witness dimension or spectrum. Pointwise lifting uses
purification uniqueness and the semilocalization mechanism of Eggeling,
Schlingemann and Werner \cite[section III]{Eggeling2002}. After physical flags convexify gain
functions, the uniform step is an instance of classical minimax, as in Sion
\cite[Theorem 4.2 (Kneser--Fan), p.~175]{Sion1958}. The compact side is the
input state space, and evaluation is concave in the input and affine in the
gain function. The finite separation proof above makes this application
explicit. The error conversion and limit order then give the collision
interpretation needed here.

\section{Active repair and exactification}
\label{sec:repair}

\begin{theorem}[active support repair and exactification]
\label{thm:active-repair}
For a fixed $\Phi$ with Choi rank $k$, every sufficiently $\eta$-close finite
collision $(U,\tau)$ on $R$ bath dimensions admits a support-repaired
collision on $(k+1)R$ dimensions with the same positive initial spectrum, Choi
support contained in that of $\Phi$, channel error at most
$\eta+\varepsilon_\eta$, and gain increase at most
$\beta_q(\varepsilon_\eta)$, where
$\varepsilon_\eta=O_\Phi(\eta^{1/6})$. A further counted flagged correction
implements $\Phi$ exactly, on at most $(k+1)R+k$ dimensions, with gain
increase tending to zero uniformly in $R$ and $\tau$. Combining this repair
with Theorem~\ref{thm:uniform-gain} gives $\kexact=\kappa$. The small-error
threshold and constants may depend on the fixed $\Phi$ (including $k$,
$C_\Phi$ and its smallest positive Choi eigenvalue $\mu$), but not on $R$,
$\tau$ or the auxiliary dimension.
\end{theorem}

\begin{lemma}[amplified Gram completion]
\label{lem:gram-completion}
The fixed Kraus span admits the following almost-isometry completion with an
auxiliary-dimension-independent constant. Uniform amplification bounds for maps
into a fixed matrix algebra are standard, including Smith's lemma
\cite{Smith1983}. We give an explicit trace-slice bound and use it to complete
the Gram defect within the prescribed Kraus span. Fix a minimal Kraus family $A_a$ for
$\Phi$, $\mathcal{K}=\operatorname{span}A_a$ and
$\mathcal{D}=\operatorname{span}A_a^*A_b$. The Gram map $L(E_{ab})=A_a^*A_b$
maps $\one_k$ to $\one_q$. A Hermiticity-preserving right inverse $\Rmap$ on
$\mathcal{D}$ can be expressed as a finite sum of trace functionals times
fixed Hermitian matrices. The slice inequality bounds every amplification
$\Rmap\otimes\id$ by one constant $C_\Phi$ independent of the auxiliary
dimension.
\end{lemma}

\begin{proof}
Choose Hermitian bases $H_l$ of $\mathcal{D}$, Hermitian preimages $C_l$, and
trace-duals $F_l$. Then $\Rmap(X)=\sum_l \Tr(F_l X)\,C_l$ and one may take
\begin{equation}
\label{eq:gram-completion}
C_\Phi=\max\Bigl(1,\ \sum_l \lVert F_l\rVert_1\,\lVert C_l\rVert_\infty\Bigr).
\end{equation}
Diagonalizing $F_l$ expresses each amplified slice as a signed sum of diagonal
compressions of $X$; each compression has norm at most $\lVert X\rVert$. This
proves the bound at every dimension. Positivity of $\Rmap$ is neither asserted
nor needed.

For $V=\sum_a A_a\otimes B_a$ with $\Delta=V^*V-\one$ and
$\lVert\Delta\rVert\le d$, set $t=C_\Phi d$ and
$G=(t\one-(\Rmap\otimes\id)(\Delta))/(1+t)\ge 0$. Its square-root block
columns $C_a$ obey $C_a^*C_b=G_{ab}$. Thus $W=\sum_a A_a\otimes C_a$ has
$W^*W=(L\otimes\id)(G)=(t\one-\Delta)/(1+t)$. The stacked map
$[V/\sqrt{1+t};\,W]$ is therefore an exact isometry whose coefficients stay in
$\mathcal{K}$. Polar normalization need not preserve $\mathcal{K}$ and is not
used. Here the block columns $C_a$ map the input bath $\rg{H}_0$ into
$\CC^k\rg{H}_0$, and $\lVert W\rVert\le\sqrt{(t+d)/(1+t)}$.
\end{proof}

\begin{proof}[Proof of Theorem~\ref{thm:active-repair}]
The coefficient projection below is an offline matrix construction; the device
does not measure a Choi projector. Throughout this proof $\epsilon$ is a
cutoff parameter, distinct from the smoothing error $\varepsilon$ of
Definition~\ref{def:intrinsic-costs}.

For an $\eta$-close mixed collision $U,\tau$ implementing $\Psi$, choose a
bath eigenbasis $\tau=\sum_j\lambda_j\ketbra{j}{j}$ and write
$U=\sum_{b,j}U_{bj}\otimes\ketbra{b}{j}$. If $\Pi_{\mathcal{K}}$ is
Hilbert--Schmidt projection of system matrices onto $\mathcal{K}$, set
$V=\sum_{b,j}\Pi_{\mathcal{K}}(U_{bj})\otimes\ketbra{b}{j}$ and $L_0=U-V$. The
channel Kraus operators are $\sqrt{\lambda_j}\,U_{bj}$; the projection itself
is unweighted and the spectrum enters their Choi sum through $\lambda_j$. Let
$P_{\mathcal{K}}$ be the projector onto the vectorized Kraus span, the support
of $J_\Phi$. Normalized Choi leakage is
$p=\Tr((\one-P_{\mathcal{K}})J_\Psi)=\Tr((\one\otimes\tau)L_0^*L_0)/q\le\eta$.
Set $E=\Tr_{\rg{S}}L_0^*L_0$. The positive-operator inequality
$X\le q\,\one\otimes\Tr_{\rg{S}}X$ and the cutoff
$P=\mathbf{1}_{[0,\epsilon^2/q]}(E)$ give
\begin{equation}
\label{eq:cutoff}
\lVert L_0(\one\otimes P)\rVert\le\epsilon,\qquad
\Tr(\tau(\one-P))\le q^2\eta/\epsilon^2 .
\end{equation}
For completeness the domination follows by applying Cauchy--Schwarz to the $q$
vectors $X^{1/2}(\ket{i}\xi_i)$, then bounding each diagonal block $X_{ii}$ by
their sum. The cutoff-weight bound follows from
$E\ge(\epsilon^2/q)(\one-P)$ and $\Tr(\tau E)=qp$. Neither estimate assumes
$[\tau,P]=0$.

Repair $V$ on $\rg{S}P\rg{B}$ with the Gram construction, and use the target's
pure Stinespring map on $\rg{S}(\one-P)\rg{B}$, retaining the complete bad-bath
input via $\ket{\psi}\ket{b}\mapsto\sum_a A_a\ket{\psi}\ket{a}\ket{b}$ in the
sector $\CC^k(\one-P)\rg{B}$, orthogonal to $\rg{B}$ and $\CC^k P\rg{B}$ used
by the good map. The resulting isometry has bath output
$\hat{\rg{B}}=\rg{B}\oplus(\CC^k\otimes\rg{B})$, of dimension $(k+1)R$.
Initialize $\hat\tau=\tau\oplus 0$ and complete the initialized columns to an
active-only unitary. The old positive spectrum is unchanged, the old $\rg{F}$
remains inert, and the system channel's Choi support lies in the target face.
Off-diagonal cutoff coherences are retained in the global output; no pinching,
measurement or discard occurs.

With $d_\epsilon=2\epsilon+\epsilon^2$, $t_\epsilon=C_\Phi d_\epsilon$ and
$a_\epsilon=\epsilon+t_\epsilon/2+\sqrt{(t_\epsilon+d_\epsilon)/(1+t_\epsilon)}$,
the joint purified output distance is at most
$a_\epsilon+2q\sqrt{\eta}/\epsilon$, uniformly over inputs. Here
$\lVert\Delta\rVert\le d_\epsilon$ on the good sector because
$V(\one\otimes P)=U(\one\otimes P)-L_0(\one\otimes P)$ with $U$ an isometry
and $\lVert L_0(\one\otimes P)\rVert\le\epsilon$, so
$\lVert(\one\otimes P)(V^*V-\one)(\one\otimes P)\rVert\le 2\epsilon+\epsilon^2$.
The good-input operator distance is at most $a_\epsilon$: on that sector
$V/\sqrt{1+t_\epsilon}-U=(V-U)/\sqrt{1+t_\epsilon}+(1/\sqrt{1+t_\epsilon}-1)U$
has norm at most $\epsilon+t_\epsilon/2$, since $U$ is an isometry and
$1-1/\sqrt{1+t}\le t/2$, while the $W$ block contributes at most
$\sqrt{(t_\epsilon+d_\epsilon)/(1+t_\epsilon)}$. Taking
$\epsilon=\eta^{1/3}$ gives $\varepsilon_\eta=O_\Phi(\eta^{1/6})$. The bad
component of the same initial purification has vector norm at most
$q\sqrt{\eta}/\epsilon$; two isometries differ on it by at most twice that
norm. Pure-state half
distance is bounded by vector distance. Thus cutoff coherences are controlled
without removing them. At $\eta=0$ retain the original collision. Both gains
equal $S(\rg{Q}\rg{A}\mid\rg{F})$ on a $q$-dimensional input purification.
Thus the conditional-entropy bound of Winter \cite[Lemma 2]{Winter2016} gives
$\sup_\rho\lvert f_{\hat C}-f_C\rvert\le\beta_q(\varepsilon_\eta)$, where
$\beta_q(\varepsilon)=4\varepsilon\log q+(1+\varepsilon)
H_2(\varepsilon/(1+\varepsilon))$, for sufficiently small
$\varepsilon_\eta\le 1$. The bound involves $\dim\rg{Q}\rg{A}=q^2$; it has no
bath or purifier dimension factor.

Let $\mu$ be the smallest positive eigenvalue of $J_\Phi$ and
$\bar\eta=\eta+\varepsilon_\eta$. Inside its support,
$\lVert J_{\hat\Psi}-J_\Phi\rVert_\infty\le 2\bar\eta$ and
$J_\Phi\ge\mu P_{\mathcal{K}}$. Therefore $f=2\bar\eta/(\mu+2\bar\eta)$
ensures $J_\Phi-(1-f)J_{\hat\Psi}\ge 0$. Its partial trace is $f\one/q$, so
division by $f$ defines a CPTP correction $\Theta$ of Kraus rank at most $k$.
A counted direct-sum collision implements $\Phi=(1-f)\hat\Psi+f\Theta$ exactly
and obeys
\begin{equation}
\label{eq:pointwise-flag}
g_{\mathrm{exact}}\le g_C+\beta_q(\varepsilon_\eta)+f\log k .
\end{equation}
Indeed $g_{\mathrm{exact}}\le(1-f)(g_C+\beta_q(\varepsilon_\eta))+f\log k$,
and $g_C\ge f_C(\pi_q)\ge 0$ by entropy conservation and subadditivity. This
uses the function-level flag identity. At zero error the correction is
omitted. Final active dimension is at most $(k+1)R+k$. If the original
positive eigenvalues are $\lambda_i$, the exact collision's eigenvalues are
$((1-f)\lambda_i,f)$, with entropy $H_2(f)+(1-f)S(\tau)$. Thus spectrum
preservation belongs to support repair; exactification changes the spectrum
explicitly. The initial entropy itself need not change uniformly little as $R$
grows. That last sentence describes the repair; it is not a gap in what the
repair hands on. Theorem~\ref{thm:causal-balancing} accepts any fixed finite
exact collision together with whatever finite spectrum that collision actually
has. It nowhere asks the exact witness's spectrum to be close to the
approximate one's, and it uses no bound on eigenvalue ratios, on bath
dimension or on concentration constants. Section~\ref{sec:region} selects a
witness and fixes all of its dimensions, spectra and constants for each $m$
before any horizon is chosen, so a spectrum that jumps at exactification costs
nothing in the diagonal that follows.

Choose a sufficiently small fixed $\eta$ for each gain slack, use
Theorem~\ref{thm:uniform-gain} to choose one finite approximate witness, and
exactify it. The uniform penalty gives $\kexact\le\kappa$; exact witnesses
belong to every feasible set of Theorem~\ref{thm:uniform-gain}, giving the
reverse inequality. No horizon is selected in this argument. An exact pure
Stinespring collision has gain $h$, so replacing any witness with gain above
$h$ by this pure collision yields exact witnesses with $0\le g\le h$ and
$g\to\kappa$. These, with their actual arbitrary finite spectra, are the
inputs Theorem~\ref{thm:causal-balancing} accepts. This closes the proof
without using Theorem~\ref{thm:causal-balancing} or a horizon limit.
\end{proof}

\begin{corollary}[pointwise smoothing and privacy funnel]
\label{cor:pointwise-funnel}
Using the lifting construction in Theorem~\ref{thm:uniform-gain} and the
pointwise repair in Theorem~\ref{thm:active-repair}, for each fixed full-rank
$\rho$ one has
\[
\lim_{\varepsilon\downarrow 0}k_{\rho,\varepsilon}=K_0(\rho),
\qquad
\kappa=\sup_{\rho>0}\bigl[S(\chi_\rho)-\Pfun^{\psi_{\rg{E}\rg{Q}}}(0)\bigr].
\]
\end{corollary}

\begin{proof}
To justify the exact mapping in Definition~\ref{def:intrinsic-costs}, purify
any exact extension on $\rg{Q}\rg{A}\rg{Z}\rg{C}$. Uniqueness of purification
gives an isometry from $\supp\psi_{\rg{E}}$ to $\rg{Z}\rg{C}$; tracing
$\rg{C}$ and extending the map arbitrarily outside that support gives a
channel $\rg{E}\to\rg{Z}$. Conversely such a channel yields an extension,
proving the claimed parametrization even when $\psi_{\rg{E}}$ is singular.

For stability, fix $\ell=\lambda_{\min}(\rho)>0$. Any $\varepsilon$-feasible
extension lifts to a finite collision $C$ at error at most
$\eta=\varepsilon/\ell$, with its objective equal to $f_C(\rho)$. The support
repair above changes this pointwise gain by at most
$\beta_q(\varepsilon_\eta)$. Let $w_\eta$ be the final correction weight.
Function-level flag mixing, $f_\Theta(\rho)\le\log k$, and
$f_C(\rho)\ge S(\rho)-\log q\ge-\log q$ give
\[
f_{\mathrm{exact}}(\rho)\le f_C(\rho)+\beta_q(\varepsilon_\eta)
+w_\eta(\log k+\log q).
\]
The exact collision supplies an exactly feasible extension. Taking arbitrarily
small optimization slack at each positive $\varepsilon$ below the
$\Phi$-dependent threshold of Theorem~\ref{thm:active-repair}, applied at
$\eta=\varepsilon/\ell$, therefore gives
\begin{equation}
\label{eq:smoothing-squeeze}
k_{\rho,\varepsilon}\ \le\ K_0(\rho)\ \le\ k_{\rho,\varepsilon}
+\beta_q(\varepsilon_{\varepsilon/\ell})
+w_{\varepsilon/\ell}(\log k+\log q).
\end{equation}
The two penalties vanish at fixed $\Phi$ and $\rho$ independently of the
extension dimension. This proves the equality before taking the full-rank
supremum. No infimum is interchanged with a limit, and no finite
auxiliary-dimension bound, optimizer, or uniformity as $\ell$ tends to zero is
asserted. The pointwise lower bound is needed here; $g_C\ge 0$ alone would not
suffice.
\end{proof}

This reformulation follows from this paper's repair, not from the prior
privacy-funnel definition, and does not alter the order in
Definition~\ref{def:intrinsic-costs}.

\section{Causal balancing of a fixed exact collision}
\label{sec:balancing}

\begin{theorem}[fixed exact witness]
\label{thm:causal-balancing}
Fix one finite exact collision $C=(U,\tau)$ for $\Phi$, with full active cell
dimension $R$. Write $\sigma=S(\tau)$, $v=\max_\rho S(\Gamma_C(\rho))$,
$g=v-\sigma$ and suppose $g\le h$. Then an all-horizon family in
Definition~\ref{def:closed-device} achieves vanishing complete adaptive error
and actual rate limits
\[
r=b=\frac{h+g}{2},\qquad s=a=\frac{h-g}{2}.
\]
Its complete initial positive spectrum is flat. More specifically it has a
pure clock, a maximally mixed seed of dimension $J_T$, $K_T$ mixed data
qubits, and pure remaining data, with $\log J_T/T\to 0$ and $K_T/T\to a$
separately. No uniform bound on the fixed witness's dimension, eigenvalue
ratios, concentration constants or circuit complexity is imposed.
\end{theorem}

\begin{proof}
Feeding $\pi_q$ and using total entropy conservation gives $v\ge\sigma$. For a
purified input, $\rg{B}\rg{F}$ is a complementary output and $\rg{F}$ has
entropy $\sigma$; hence $h\le v+\sigma$. Therefore $0\le a\le\sigma$, $b=a+g$
and $v\le\log R$. These are the only resource inequalities used to allocate
the construction.

Appendix~\ref{app:balancing} supplies the complete proof.
Lemma~\ref{lem:adaptive-projector} gives uniform adaptive spectral
concentration. Lemma~\ref{lem:exterior-cuts} transfers the virtual
minimal-environment support to a purification of the complete exterior and
retains the normalization in the simultaneous projector bound.
Lemma~\ref{lem:seed-decoupling} proves the second moment with the physical
seed still present. Lemma~\ref{lem:integer-wires} constructs the nonflat cell
marginal reversibly from flat stock on exact integer wires.
Proposition~\ref{prop:joint-recycling} proves joint recycling and every
prefix's error with disjoint parked residues.
Proposition~\ref{prop:actual-spectrum} counts the clock, actual spectrum and
every-horizon limit. All are proved in this paper.
\end{proof}

The encoder is the sender's unitary from fully quantum Slepian--Wolf:
sender $=\rg{C}$, retained share $=\rg{M}$, message $=\rg{G}$, receiver
$=\rg{N}$ and reference $=\rg{H}$; the receiver is written $\rg{N}$ here to
keep $\rg{F}$ for the purifiers of Definition~\ref{def:closed-device} and
Appendix~\ref{app:balancing}. Its decoder acts on $\rg{G}\rg{N}$, so omitting
that decoder leaves $\rg{M}\rg{H}$ unchanged. The physical device
parks and counts $\rg{G}$. The split and its half-sum rates are established
prior ingredients, specifically Abeyesinghe, Devetak, Hayden and Winter
\cite[Theorems IV.1--IV.2, Lemma IV.5 and Eq.~(30), pp.~6--9, with the
asymptotic rate choice following Eq.~(35) in section VII]{Abeyesinghe2009}.
The present appendix supplies the adaptive spectral control, reversible
flat-stock preparation and counted reuse of seed and workspace needed to turn
this split into the stated closed device. These implementation steps are part
of the operational synthesis.

\section{Assembling the region}
\label{sec:region}

\begin{proof}[Proof of Theorem~\ref{thm:rate-region}]
The argument uses Theorems~\ref{thm:same-spectrum-converse},
\ref{thm:uniform-gain}, \ref{thm:active-repair} and
\ref{thm:causal-balancing}.

First apply the converse bounds of Theorem~\ref{thm:same-spectrum-converse} to
an arbitrary all-horizon family, writing $x_T=\log R_T/T$ and
$y_T=S(\omega_T)/T$. For each $\zeta>0$, eventually \emph{both}
$x_T-y_T\ge\kappa-\zeta$ and $x_T+y_T\ge h-\zeta$. Their common actual
initialization therefore gives $2x_T\ge h+\kappa-2\zeta$, regardless of
whether $y_T$ converges. Thus $\liminf x_T\ge(h+\kappa)/2$ and $\cseq$ has
this lower bound. For rate-limit pairs, $s\ge 0$, $r+s\ge h$ and
$r-s\ge\kappa$ follow and survive closure. The definition gives
$0\le\kappa\le h$: $S(\rg{Q}\rg{A}\mid\rg{Z})\ge S(\rho)-\log q$ is
nonnegative at $\rho=\one/q$, while the exact trivial-$\rg{Z}$ extension has
entropy $S(\chi_\rho)\le h$.

Theorems~\ref{thm:uniform-gain} and \ref{thm:active-repair} give
$\kexact=\kappa$. For each $m\ge 1$ select a finite \emph{exact} collision
with gain $<\kappa+1/m$. Replace it by a pure exact Stinespring collision of
gain $h$ if its gain exceeds $h$. Then $\kappa\le g_m\le h$ and
$0\le g_m-\kappa<1/m$. This also handles $\kappa=h$ without assuming
attainment. Set $a_m=(h-g_m)/2$ and $b_m=(h+g_m)/2$. All dimensions, spectra
and concentration constants are fixed separately for each $m$.

Theorem~\ref{thm:causal-balancing} supplies, for each fixed witness,
all-horizon devices with error $\delta_{m,T}\to 0$ and rates $(b_m,a_m)$.
Their exact initialization is
\begin{equation}
\label{eq:rate-diagonal}
\begin{aligned}
\omega_{m,T}&=\ketbra{0}{0}_{\mathrm{clock}}\otimes\pi_{J_{m,T}}
\otimes\pi_{2^{K_{m,T}}}\otimes\ketbra{0}{0}_{\mathrm{rest}},\\
R_{m,T}&=(T+1)\,J_{m,T}\,2^{A_{m,T}},\qquad
S(\omega_{m,T})=\log J_{m,T}+K_{m,T}.
\end{aligned}
\end{equation}
Here $0\le K_{m,T}\le A_{m,T}$ are integers. Theorem~\ref{thm:causal-balancing}
proves separately that $\log J_{m,T}/T\to 0$ and $K_{m,T}/T\to a_m$. Choose a
finite $H_m$ such that for \emph{every} $T\ge H_m$, the error and each of
\[
\Bigl\lvert\frac{\log R_{m,T}}{T}-b_m\Bigr\rvert,\quad
\Bigl\lvert\frac{S(\omega_{m,T})}{T}-a_m\Bigr\rvert,\quad
\frac{\log J_{m,T}}{T},\quad
\Bigl\lvert\frac{K_{m,T}}{T}-a_m\Bigr\rvert
\]
are at most $1/m$. Take $T_1=\max(1,H_1)$ and
$T_{m+1}=\max(H_{m+1},T_m+1)$. For $T\ge T_1$ select the complete $m$-th
device with $m=m(T)=\max\{j: T_j\le T\}$; use exact pure preallocated cells
for earlier $T$. Since $T_m\ge m$, this maximum is finite; for every $M$,
$T\ge T_M$ implies $m(T)\ge M$. Thus $m(T)\to\infty$ over all horizons. Every
selected $T$ lies beyond $H_m$, so the bounds hold throughout
$T_m\le T<T_{m+1}$. With
\begin{equation}
\label{eq:corner}
r_0=\frac{h+\kappa}{2},\qquad s_0=\frac{h-\kappa}{2},
\end{equation}
the actual rate errors are at most $3/(2m(T))$, the adaptive error is at most
$1/m(T)$, $K_T/T\to s_0$, and $\log J_T/T\to 0$. The separate seed threshold
is needed for the region proof; it must not be inferred merely from total-rate
limits. Hardware and its actual initializer are selected before the user for
each $T$; no witness changes online. Each elementary collision implements
\emph{exact} $\Phi$, so there is no accumulated $T\eta$ term. This proves
$\cseq=r_0$ with actual entropy rate $s_0$. No effective thresholds or
witness-size bound are claimed for this achievability half. The converse half
is the half that constrains an actual device: for every finite horizon,
Theorem~\ref{thm:same-spectrum-converse}'s two inequalities hold with explicit
constants in terms of that horizon's own error $\delta_T$, although the
extension-cost term $k_{\rho,\delta_T}$ itself is an infimum that is not
asserted to be computable. Only the thresholds $H_m$ above are non-effective.

The selected base family has $P_T=J_T 2^{K_T}$ equal positive eigenvalues
$1/P_T$ and $R_T-P_T$ zeros. To reach $0\le s\le s_0$ put $d=s_0-s$ and
$l_T=\min(K_T,\lfloor dT\rfloor)$. Then $l_T/T\to d$, including the clipped
endpoint $s=0$. Replace $l_T$ of the initially mixed data qubits by halves of
Bell pairs with $l_T$ \emph{new}, counted active partner qubits $\rg{P}$. Keep
the other initial factors unchanged and extend the repeated unitary to
$W_T\otimes\one_{\rg{P}}$. The new old-bath marginal is \emph{exactly}
$\omega_T$. Partial trace over $\rg{P}$ commutes with every device and
adaptive-user operation, so the entire original marginal process and its
complete adaptive error are unchanged. This uses a new offline initializer,
not an operation on the original inaccessible purifier.

Exactly, the new dimension is $R_T 2^{l_T}$, its positive rank is
$P_T^-=J_T 2^{K_T-l_T}$, and its positive spectrum is $1/P_T^-$ repeated
$P_T^-$ times, with zeros elsewhere. Its entropy is $S(\omega_T)-l_T$.
Therefore the limiting pair is $(r_0+d,s_0-d)=(h-s,s)$. At $s=0$ the remaining
entropy is $\log J_T+\max(K_T-\lfloor s_0T\rfloor,0)=o(T)$; a zero entropy
\emph{rate} does not require exactly pure finite-horizon initialization.
Partners remain counted and untouched, and a purifier of the new initializer
remains inaccessible.

For $s\ge s_0$ instead append $n_T=\lfloor(s-s_0)T\rfloor$ unused maximally
mixed qubits in product, extending $W_T$ by identity. Dimension and positive
rank both multiply by $2^{n_T}$ and entropy increases by $n_T$. The positive
spectrum remains flat and the user process is unchanged. This gives
$(\kappa+s,s)$. For any $r\ge r_{\min}(s)=\max(h-s,\kappa+s)$, append
$\lfloor(r-r_{\min}(s))T\rfloor$ unused pure qubits. This increases dimension,
leaves positive rank and entropy unchanged, and adds only zero eigenvalues.
All rounding errors vanish after division by $T$.

Every finite rate pair in the stated region is thus realized with actual rate
limits; the set is already closed. At each finite horizon,
$\rank\le\text{dimension}$ and $S\le\log R$ hold: internal purification
decreases rank and increases dimension, mixed padding multiplies both equally,
and pure padding only increases dimension. Asymptotically $r\ge s$ follows
from $\kappa\ge 0$. At $q=1$ the trivial bath-free device and padding realize
$r\ge s\ge 0$. If $\kappa=h$ then $s_0=0$; if $\kappa=0$ then $r_0=s_0=h/2$.
Both endpoints are covered by the same integer constructions.

The dependency order is acyclic: Theorem~\ref{thm:same-spectrum-converse} is
an independent converse; the fixed-error minimax of
Theorem~\ref{thm:uniform-gain} precedes the exactification of
Theorem~\ref{thm:active-repair}; Theorem~\ref{thm:causal-balancing} acts on
one fixed exact witness without Theorems~\ref{thm:uniform-gain} and
\ref{thm:active-repair}; the all-horizon diagonal then combines these inputs.
Internal purification and padding are last. The encoder-only Slepian--Wolf
attribution of Section~\ref{sec:balancing} remains in force.
\end{proof}

\section{Examples and limitations}
\label{sec:examples}

\begin{proposition}[unitary channels]
\label{prop:unitary}
For $\Phi(\rho)=V\rho V^*$, $h=\kappa=0$ and the minimum memory rate is zero.
\end{proposition}

\begin{proof}
A minimal environment is one-dimensional, so $h=0$. Every feasible extension
has $S(\rg{Q}\rg{A}\mid\rg{Z})\ge S(\rho)-\log q$, which is nonnegative at
$\rho=\pi_q$. Taking a trivial $\rg{Z}$ gives $\kappa\le h$, hence $\kappa=0$.
The bath-free repeated $V$ implements all visits exactly and has both rates
zero.
\end{proof}

\begin{proposition}[qubit dephasing]
\label{prop:dephasing}
For $\Delta(\rho)=(\rho+Z\rho Z)/2$, $h=1$, $\kappa=0$, and
Theorem~\ref{thm:rate-region} gives the corner $r=s=1/2$.
\end{proposition}

\begin{proof}
A minimal dilation copies the computational-basis label into a qubit
environment. Its output entropy is the binary entropy of the input diagonal,
maximized at one. A fair classical bath bit controlling $\one$ or $Z$
implements $\Delta$. Its active marginal stays $\pi_2$ for every input, so
$g=0$. Theorem~\ref{thm:uniform-gain} and $\kappa\ge 0$ imply $\kappa=0$.
\end{proof}

There is also a direct exact streaming realization of the corner. Put
$m=\lceil T/2\rceil$, initialize $m$ bath qubits in $\pi_{2^m}$, and choose
$T$ distinct binary-independent Pauli generators from
$X_1,Z_1,\dots,X_m,Z_m$. At visit $t$ use the system computational bit to
control the $t$-th generator $P_t$ on the bath. For a history
$x\in\{0,1\}^T$, its ordered bath word is
$P(x)=P_T^{x_T}\cdots P_1^{x_1}$. Distinct histories give distinct Pauli words
up to phase and $\Tr(P(x)P(y)^*)/2^m=\delta_{x,y}$. Expanding any purified
adaptive tester in computational histories, its retained user vectors are
weighted by precisely this Gram matrix after tracing the bath. All
cross-history terms vanish, exactly as for fresh dephasing records; this
proves complete adaptive service, including reference-entangled inputs. A pure
$(T+1)$-state clock compiles these gates into one repeated unitary as in
\eqref{eq:repeated-unitary}. Thus $\log R=m+\log(T+1)$ and $S(\omega)=m$, with
no discard. This is the counted streaming specialization of the
orthogonal-unitary construction in Boes et al.\ \cite[Lemma 1]{Boes2018}, not a claim to
originate its square-root dephasing memory saving. The subsequent 2020
erratum corrects section V / Theorem 3's expander result, which is not used
here; it does not modify the dephasing lemma cited above.

\begin{proposition}[pure replacement]
\label{prop:pure-replacer}
For $\Phi(\rho)=\ketbra{0}{0}\Tr\rho$ on $\CC^q$, $h=\kappa=\log q$ and
$\cseq=\log q$.
\end{proposition}

\begin{proof}
A minimal dilation maps $\ket{\psi}$ to $\ket{0}_{\rg{A}}\ket{\psi}_{\rg{E}}$,
so $h=\log q$. For any feasible extension,
$S(\rg{Q}\rg{A}\mid\rg{Z})\ge S(\rho)-S(\sigma_{\rg{A}})$. Its $\rg{A}$
marginal is within $\varepsilon$ of $\ketbra{0}{0}$, hence
$S(\sigma_{\rg{A}})\to 0$ at fixed $q$ by entropy continuity. At $\rho=\pi_q$
this gives $\kappa\ge\log q$, while $\kappa\le h$ gives equality.
Specializing the preallocated cell construction of Ryb\'ar and Ziman
\cite[section III]{RybarZiman2008}, preallocate
$T$ pure $q$-dimensional cells and swap each arriving input into its
designated cell, returning $\ket{0}$; retain all cells and count the pure
clock. This exact device has $\log R=T\log q+\log(T+1)$ and $S(\omega)=0$. Its
bath grows with $T$, as it must: Theorem 2 of Ryb\'ar and Ziman excludes a
fixed finite memory for this nonunital channel. No mixed-replacement or other
companion classification is needed for this example.
\end{proof}

\begin{proposition}[a lower bound on $\kappa$ over the input space alone]
\label{prop:kappa-lower}
For every channel $\Phi$,
\[
\kappa(\Phi)\ \ge\ \max_\rho\ \bigl[S(\rho)-S(\Phi(\rho))\bigr].
\]
\end{proposition}

\begin{proof}
Fix $\rho>0$ and any $\varepsilon$-feasible extension $\sigma$ in
Definition~\ref{def:intrinsic-costs}, with $q\ge 2$. The chain rule and the
conditional form of Araki--Lieb give
\[
S(\rg{Q}\rg{A}\mid\rg{Z})=S(\rg{Q}\mid\rg{Z})+S(\rg{A}\mid\rg{Q}\rg{Z})
\ \ge\ S(\rho)-S(\sigma_{\rg{A}}),
\]
using $S(\rg{Q}\mid\rg{Z})=S(\rho)$ from the exact product
$\sigma_{\rg{Q}\rg{Z}}=\rho^{T}\otimes\sigma_{\rg{Z}}$ and
$S(\rg{A}\mid\rg{Q}\rg{Z})\ge-S(\rg{A})$. Since
$\Dist(\sigma_{\rg{Q}\rg{A}},\chi_\rho)\le\varepsilon$ and
$\Tr_{\rg{Q}}\chi_\rho=\Phi(\rho)$, the marginal $\sigma_{\rg{A}}$ is within
half distance $\varepsilon$ of $\Phi(\rho)$, so Audenaert
\cite[Theorem 1]{Audenaert2007} gives
$S(\sigma_{\rg{A}})\le S(\Phi(\rho))+\varepsilon\log(q-1)+H_2(\varepsilon)$.
Hence $k_{\rho,\varepsilon}\ge S(\rho)-S(\Phi(\rho))-\varepsilon\log(q-1)
-H_2(\varepsilon)$. Letting $\varepsilon\downarrow 0$ at fixed $\rho$ and then
taking the supremum over full-rank $\rho$ gives
$\kappa\ge\sup_{\rho>0}[S(\rho)-S(\Phi(\rho))]$, which equals the maximum over
all states by continuity of the bracket and density of the full-rank states;
the maximum exists by compactness. This uses neither
Theorem~\ref{thm:uniform-gain} nor Theorem~\ref{thm:active-repair}.

The same bound holds for every finite \emph{exact} collision $C=(U,\tau)$, by
the entropy-conservation and subadditivity budget of Ryb\'ar and Ziman
\cite[section III, Eq.~(3.7)]{RybarZiman2008}: the state
$U(\rho\otimes\tau)U^*$ on $\rg{A}\rg{B}$ is unitarily equivalent to
$\rho\otimes\tau$, so $S(\rg{A}\rg{B})=S(\rho)+S(\tau)$, while subadditivity
gives $S(\rg{A}\rg{B})\le S(\Phi(\rho))+S(\Gamma_C(\rho))$. Subtracting
$S(\tau)$,
\[
f_C(\rho)\ \ge\ S(\rho)-S(\Phi(\rho))\quad\text{for every }\rho,
\]
hence $g_C\ge\max_\rho[S(\rho)-S(\Phi(\rho))]$. This second route bounds
$\kexact$ directly and reaches $\kappa$ only through
Theorems~\ref{thm:uniform-gain} and \ref{thm:active-repair}.
\end{proof}

Unlike Definition~\ref{def:intrinsic-costs}, the right-hand side is an
optimization over the $q$-dimensional state space alone: no auxiliary system,
no smoothing and no order of limits enter it. It is exactly the universal
parallel work rate at trivial Hamiltonians of Faist, Berta and Brand\~ao
\cite[Theorem 5.1 and Eq.~(5.4)]{Faist2019}, so the entropy-deficit rate of a
closed device is at least that thermodynamic work rate;
Section~\ref{sec:related} compares the two contracts. The bound reproduces the
three examples above exactly: it is zero for a unitary channel, zero for
dephasing because a unital channel never decreases entropy, and $\log q$ for
the pure replacer.

\begin{proposition}[an interior corner]
\label{prop:interior}
For $0\le\nu,p\le 1$ let $\Phi_{\nu,p}$ on $q=2$ have the generalized
amplitude damping operators
\[
\begin{gathered}
K_0=\sqrt{p}\,\mathrm{diag}(1,\sqrt{1-\nu}),\qquad
K_1=\sqrt{p\nu}\,\ketbra{0}{1},\\
K_2=\sqrt{1-p}\,\mathrm{diag}(\sqrt{1-\nu},1),\qquad
K_3=\sqrt{(1-p)\nu}\,\ketbra{1}{0},
\end{gathered}
\]
which satisfy $\sum_a K_a^*K_a=\one$; the damping parameter is written $\nu$
to keep $\gamma$ for the concentration slack of Appendix~\ref{app:balancing}.
At $\nu=1/2$ and $p=3/4$,
\[
0<H_2(7/12)-H_2(1/3)\ \le\ \kappa\ \le\ 1-H_2(3/4)\ <\ H_2(3/8)\ \le\ h,
\]
that is $0.061572\le\kappa\le 0.188722$ and $h\ge 0.954434$, each decimal
rounded towards the side that keeps the inequality valid. Hence
$0<\kappa<h$.
\end{proposition}

\begin{proof}
Trace preservation is the diagonal identity
$p+(1-p)(1-\nu)+(1-p)\nu=1$ and $p(1-\nu)+p\nu+(1-p)=1$. For the lower
bound, the channel acts on diagonal inputs by
$\Phi(\mathrm{diag}(a,1-a))=\mathrm{diag}(3/8+a/2,\,5/8-a/2)$, so
$\rho=\mathrm{diag}(7/12,5/12)$ has output exactly $\mathrm{diag}(2/3,1/3)$,
and Proposition~\ref{prop:kappa-lower} at this $\rho$ gives
$\kappa\ge H_2(7/12)-H_2(1/3)>0$ through the entropy-decrease budget credited
there. The maximally mixed input, with $\Phi(\one/2)=\mathrm{diag}(5/8,3/8)$,
gives only $1-H_2(5/8)=0.04556\ldots$ and is not the maximiser of the bracket
in Proposition~\ref{prop:kappa-lower}; N.~Mghirbi (private communication,
2026) observed that $\mathrm{diag}(3/5,2/5)$ already improves it to
$0.061214$. A numerical search over qubit inputs, rerun by the script named
in Appendix~\ref{app:diagnostics}, finds the largest value $0.061598$ at
$\mathrm{diag}(0.58671,0.41329)$, so the rational witness above is within
$0.00003$ of the largest value found numerically; the search is a diagonal
grid with random general inputs and certifies no global maximum.

For the upper bound, exhibit an exact collision. Take $\rg{B}=\CC^2$,
$\tau=\mathrm{diag}(p,1-p)$, and the unitary $U$ fixing $\ket{00}$ and
$\ket{11}$ and acting on the ordered pair $\ket{01},\ket{10}$ by
\[
\begin{pmatrix}\sqrt{1-\nu} & -i\sqrt{\nu}\\
-i\sqrt{\nu} & \sqrt{1-\nu}\end{pmatrix}.
\]
That block is unitary, so $U$ is. Writing
$U=\sum_{b,j}U_{bj}\otimes\ketbra{b}{j}$ as in Section~\ref{sec:repair}, its
weighted bath matrix elements $\sqrt{\lambda_j}\,U_{bj}$ come from
$U_{00}=\mathrm{diag}(1,\sqrt{1-\nu})$,
$U_{11}=\mathrm{diag}(\sqrt{1-\nu},1)$,
$U_{10}=-i\sqrt{\nu}\ketbra{0}{1}$ and
$U_{01}=-i\sqrt{\nu}\ketbra{1}{0}$, weighted by $\sqrt{p}$ and
$\sqrt{1-p}$; these are the four operators above up to individual phases, so
$\Phi_C=\Phi$ exactly. This is a lawful finite collision in the sense of
Definition~\ref{def:collision}: $\tau$ is user independent, $U$ acts on system
and active bath only, and $\log\dim\rg{B}=1$ with $S(\tau)=H_2(p)$. Because
$\dim\rg{B}=2$, its active output entropy is at most one for \emph{every}
input, whatever the off-diagonal entries, so $g_C\le 1-H_2(3/4)$. This bound
is attained: the input $\mathrm{diag}(1/4,3/4)$ sends the active bath output
of this collision to $\one/2$, so $g_C=1-H_2(3/4)$ exactly and the witness
cannot be improved. Theorems~\ref{thm:uniform-gain} and
\ref{thm:active-repair} give $\kappa=\kexact\le g_C$.

For $h$, the minimal dilation of a pure input has system and environment with
the same nonzero spectrum, so $h\ge S(\Phi(\ketbra{1}{1}))=H_2(3/8)$, since
$\Phi(\ketbra{1}{1})=\mathrm{diag}(3/8,5/8)$. A mixed input does better: at
$\mathrm{diag}(1/3,2/3)$ the complementary output has entropy $1.148402$ (also
observed by N.~Mghirbi; the largest value found numerically is $1.148986$ at
$\mathrm{diag}(0.35586,0.64414)$), so in fact $h>1$, although the analytic
bound suffices here. Strictness follows from $H_2$
being strictly increasing on $(0,1/2)$:
$H_2(3/8)>H_2(1/4)=H_2(3/4)>1/2>1-H_2(3/4)$.
\end{proof}

This locates the corner without computing $\kappa$. Two consequences are
already visible from the bracket. Since $\kappa>0$, no device for this channel
has a vanishing entropy-deficit rate, so $r=s$ is impossible and the minimum
dimension rate $(h+\kappa)/2$ strictly exceeds $h/2$. Since $\kappa<h$, the
optimal initial entropy rate $s_0=(h-\kappa)/2$ is strictly positive, so the
cheapest device is strictly mixed rather than pure. The bracket also gives
$g_C\le h$ for this witness, so Theorem~\ref{thm:causal-balancing} applies to
it directly and produces an explicit all-horizon family. The finite arithmetic
behind the displayed decimals is rerun by the script named in
Appendix~\ref{app:diagnostics}; the strict inequalities above are analytic and
use no optimizer.

The quantity $\cseq$ measures complete operating storage. It is not identified
with a universal unavoidable terminal residual rate under unrestricted growth
of a working bank. A construction's parked-residue rate does not establish
such an equality. The theorem asserts no efficient witness search, circuit
bound, bath return, independent return, spatial additivity, or service law for
arbitrary multi-visit processes. It concerns one fixed memoryless channel.

\section{Related work and distinct zero-cost statements}
\label{sec:related}

Ryb\'ar and Ziman's \emph{Repeatable quantum memory channels}
\cite[sections II--IV]{RybarZiman2008} is a close physical comparison. After
Eq.~(2.3) they assume product inputs. Their section III repeatability
definition fixes one memory, initializer and unitary for all repetitions and
requires the same single-use marginal channel at every step. Theorem 1
supplies finite random-unitary repeaters; Theorem 2 rules out nonunital
finite-memory repeaters. Their conclusion leaves general unital,
non-random-unitary repeatability open and acknowledges correlations between
outputs and effects of measurements. These are statements of that 2008 paper,
not a claim about the current status of its question.

\begin{proposition}[shared-seed separation]
\label{prop:shared-seed}
One fair bath bit controlling $\one$ or $Z$ on every qubit is such a marginal
dephasing repeater, but its complete service error against independent
dephasing is at least $1-2^{1-T}$.
\end{proposition}

\begin{proof}
For any product input list, the bath bit's probabilities remain $(1/2,1/2)$,
so each individual output has dephasing's marginal. For $T$ inputs $\ket{+}$,
the joint output is
$(\ketbra{+}{+}^{\otimes T}+\ketbra{-}{-}^{\otimes T})/2$. Independent
dephasing produces $\pi_{2^T}$. Both states are diagonal in the $X$ basis.
Summing the two occupied-string discrepancies and the other $2^T-2$
discrepancies gives half distance $1-2^{1-T}$. The event that all $X$ outcomes
agree attains it, already $1/2$ for $T=2$. No feedback is needed to
distinguish the contracts. In particular, the finite shared seed is not a
small-error device for Definition~\ref{def:complete-error-rates}.
\end{proof}

Our device may depend on $T$ and has vanishing complete adaptive error. Those
quantifiers and its joint service requirement differ from fixed finite exact
marginal repeatability. Finite $n$-repeatability and the $n$-cell permutation
construction are also in Ryb\'ar--Ziman; horizon dependence alone is not a new
distinction. Theorem~\ref{thm:rate-region} is \emph{not} asserted
to resolve Ryb\'ar--Ziman's question or an unspecified asymptotic adaptive
version.

The resource distinctions can be stated precisely from the core theorem.

\begin{center}
\small
\begin{tabular}{@{}>{\raggedright\arraybackslash}p{4.0cm}>{\raggedright\arraybackslash}p{10.4cm}@{}}
\hline
\textbf{Statement} & \textbf{Meaning and consequence}\\
\hline
Zero initial entropy & $S(\omega_T)=0$ at a specified finite horizon; stronger
than its entropy being sublinear.\\
Zero initial entropy rate & $s=0$; the region requires $r\ge h$. Exact purity
at every horizon is not implied.\\
Zero entropy deficit & $\log R_T-S(\omega_T)=0$ at a specified horizon,
equivalently $\omega_T=\pi_{R_T}$.\\
Zero entropy-deficit rate & $r-s=0$; under Theorem~\ref{thm:rate-region} this
is possible iff $\kappa=0$, with $r=s\ge h/2$.\\
Zero memory rate & $r=0$; then $s=0$ and $h=0$. In the square-channel setting
this is possible exactly for unitary channels.\\
Exact finite tracial factorization & One exact finite weighted-tracial
dilation exists, with fixed finite algebra and weights.\\
Finite-tracial approximation closure & Arbitrarily accurate finite tracial
dilations exist; no bounded dimension or exact attained dilation follows.\\
\hline
\end{tabular}
\end{center}

For the unitary characterization in this table, if $h=0$ every complementary
output is pure. Convexity forces these outputs to be the same pure state,
since a mixture of distinct pure states has positive entropy. The dilation
therefore factors as an isometry $\rg{S}\to\rg{A}$ tensor that pure state.
Equal input and output dimensions make the isometry unitary;
Proposition~\ref{prop:unitary} gives the converse.

A finite weighted-tracial dilation means $\rg{B}=\bigoplus_j\CC^{d_j}$,
$\tau=\bigoplus_j p_j\pi_{d_j}$, and a unitary preserving these sectors. It
implements $\sum_j p_j\Tr_{\rg{B}_j}U_j(\rho\otimes\pi_{d_j})U_j^*$. Each
sector's output bath entropy is at most $\log d_j$ and $\pi_q$ attains
equality in every sector, so $g=0$ by the block entropy identity.
Theorem~\ref{thm:uniform-gain} therefore implies that membership in this
approximation closure gives $\kappa=0$. The reverse implication requires an
additional gain-to-tracial approximation argument and is not a consequence
proved in this paper.

Allowing weighted direct sums is a broader exact convention than a single
maximally mixed matrix bath. Their approximation closures coincide: for large
$D$ choose nonnegative integers $m_j$ with $\sum_j m_j d_j=D$ and $m_jd_j/D$
approaching $p_j$ (for example take $D$ along multiples of the $d_j$, and
approximate the weights rationally). Replicate $U_j$ on $m_j$ copies. The
half-diamond channel discrepancy is at most the total variation in weights.
This observation supplies approximation, not exact attainment for irrational
weights.

The distinction is substantive prior mathematics. Musat and R\o rdam
\cite[Theorem 4.1]{MusatRordam2020} give Schur channels in the
finite-factorization closure with no exact finite-dimensional factorization,
in dimensions $2n+1$ for $n\ge 5$. We cite that theorem for this distinction,
without re-proving its operator-algebraic dependency chain. Likewise Lie, Son, Boes, Ng and Wilming
\cite[published Proposition 6 and Eq.~(12)]{Lie2026} relate informational
equilibrium to weighted-tracial dilations and factorizable channels. These
structural and equilibrium ingredients are prior work. No companion
finite-tracial-closure equivalence or quantitative modulus is imported into
the results proved here.

Pointwise product-reference lifting uses the Stinespring uniqueness mechanism
in Eggeling, Schlingemann and Werner \cite[section III]{Eggeling2002}.
The uniform optimization and smoothing
order in Section~\ref{sec:collision-gain} are written separately. Ordinary
dilation continuity, as in Kretschmann, Schlingemann and Werner
\cite[Definition 1 and Theorem 1]{Kretschmann2008}, does not itself impose a
comparison acting only on the active bath with an already fixed inert purifier
and positive initial spectrum. This is a difference in the stated contract,
not proof of originality of Theorem~\ref{thm:active-repair}.

In particular, purifying a mixed initializer gives a dilation into
$\rg{A}\rg{B}\rg{F}$. The environmental comparison in ordinary Stinespring
continuity can act on the joint $\rg{B}\rg{F}$; it need not factor as an
operation on $\rg{B}$ tensored with $I_{\rg{F}}$.
Theorem~\ref{thm:active-repair} instead preserves the old positive spectrum
at the support-repair stage and completes columns inside the target Kraus
span, with a gain bound uniform in the bath size. The final flagged
exactification changes the spectrum. These are the precise extra
requirements in the comparison.

For factorization background, Haagerup and Musat
\cite[Definition 1.3 and Theorem 2.2]{HaagerupMusat2011} characterize matrix
Markov maps through unitary dilations with a tracial von Neumann algebra.
Their finite von Neumann algebra need not be finite dimensional; our device
bath must be. This is background to the closure/attainment distinction above,
not a historical claim that they originated all factorization terminology.

The privacy-funnel comparison in Definition~\ref{def:intrinsic-costs}
identifies the exact pointwise cost, with its smoothing equality supplied by
Section~\ref{sec:repair}. The classical privacy funnel originates with
Makhdoumi et al.\ \cite{Makhdoumi2014}; Datta, Hirche and Winter supply the
quantum definition used here. Other channel entropies answer different questions.
In particular, Gour and Wilde
\cite[Proposition 6, Theorem 10 and Proposition 24]{GourWilde2021} give a
channel entropy and a universal parallel channel-merging theorem. Their
protocol allows operations on both output and environment shares, free one-way
classical communication, and measures net entanglement gain. Its channel
entropy is $-\log q$ for a unitary and zero for a pure replacer; our $\kappa$
is respectively zero and $\log q$. Their universal merging result is a close
antecedent, but its resource accounting and access differ from ours.

Slepian--Wolf attribution is explicit in Section~\ref{sec:balancing}. Its
encoder and the universal merging and reverse-Shannon constructions, including Bennett et al.\ \cite[Theorem 3]{Bennett2014}, provide the relevant splitting machinery. Our
immediate-output construction runs exact collisions before encoding spent
cells. It must additionally control adaptive spectra and allocate actual flat
stock, seed, work and all residues in one closed bath; communication or net
entanglement rates alone do not count that allocation.

Baghali Khanian and Leung \cite[section II, Definition 3 and Theorems 6 and 8]{KhanianLeung2025}
strengthen the reverse-Shannon comparison by allowing a general mixed
source/reference and preserved encoder-side information, unifying feedback
and non-feedback simulation. Their input is a specified source tensor power,
and the rates charge a transmitted register and entanglement. Their assisted
optimum and unassisted bounds do not allocate all encoder/decoder environments
inside one closed sequential bath. Transferring that machinery to
Definition~\ref{def:closed-device} still requires the adaptive compression,
timing and complete storage accounting of Appendix~\ref{app:balancing}.

Devetak and Yard \cite[main region and Eqs.~(1)--(2)]{DevetakYard2008}
place fully quantum
Slepian--Wolf in state redistribution: the resources are quantum communication
and net entanglement for many identical copies of a specified state with
accessible sender and receiver shares. Those resource rates do not by
themselves price the complete retained bath of
Definition~\ref{def:closed-device}.

Faist and Renner \cite[Main Result and Eq.~(2)]{FaistRenner2018}
characterize a specified
process on a specified input, preserving its reference correlations, by
coherent relative entropy with an information battery and Gibbs-preserving
free operations. This is a work resource comparison; it does not impose our
whole-bath dimension charge or one device serving all adaptive testers.

Faist, Berta and Brand\~ao \cite[Theorem 5.1 and Eq.~(5.4)]{Faist2019}
give a universal parallel
thermodynamic implementation with work rate
$\max_\rho[S(\rho)-S(\Phi(\rho))]$ for trivial Hamiltonians. Their battery and
free-operation accounting includes a partial-trace implementation, so it does
not charge the entire retained bath of
Definition~\ref{def:closed-device}. The two models are nevertheless comparable
in one direction: Proposition~\ref{prop:kappa-lower} shows that our
entropy-deficit rate $r-s$ is at least that work rate, for every channel. The
mixed-bath model of Terhal et al.\ already imposes the unitary column
constraints; its finite single-use environment analysis does not state this
repeated rate region. Lie and Jeong \cite[Theorems 1, 5--6]{LieJeong2021} study
randomness-utilizing implementations with input-independent bath output. That
restriction is absent here, and their no-secret recovery acts on a purifier
that our device cannot access.

Quantum strategy memory in Bisio et al.\
\cite[Definitions 3--4 and Theorem 3]{Bisio2012} permits free classical memory
and local CPTP implementation; Definition~\ref{def:closed-device} charges
those physical resources.

Lie and Jeong \cite[Proposition 2, Corollary 10 and Theorem 11]{LieJeongCatalytic2021}
study implementations returning a catalyst for every input. Their entropy
bounds specialize to $s\ge h/2$ and $r+s\ge h$ for a returned catalyst implementing
the tensor-power channel. Definition~\ref{def:closed-device} imposes no such return;
its region also permits $s=0$ and $r=h$. Their catalytic block decomposition
does not identify every zero-gain collision: a SWAP with a maximally mixed bath
has $g=0$ but returns the input state to the bath. No reverse closure implication
is imported here.

Kotowski and Kotowski \cite{Kotowski2026} implement a unital channel on a
$d$-dimensional system with a fresh ancilla of dimension $k$ and success
probability of order $k/\log d$, shown optimal up to constants, and simulate
highly noncommutative channels with one auxiliary qubit. Their protocol may
fail with a flag and draws a fresh ancilla at every use;
Definition~\ref{def:closed-device} removes both permissions, since the device
must succeed at every visit with one bath that is never replenished, so
neither cost measure bounds the other and nothing is imported.

Dimension and entropy need not have the same minimizer in quantum models
of classical stochastic processes, as shown by Loomis and Crutchfield
\cite{Loomis2019} and Liu et al.\ \cite{Liu2019}. Those costs concern a stationary
information-bearing memory; here actual initial entropy is priced together
with the entire closed device. In the quantum-input setting, Chang, Berk and Gu
\cite[Theorem 3 and Appendix B.5]{Chang2026} bound stationary recurrent memory
by temporal excess entropy. Their bound vanishes for a memoryless channel,
and their recurrent step may be an arbitrary CPTP operation.
Definition~\ref{def:closed-device} instead counts the entire closed apparatus
needed to supply repeated calls. These comparisons explain the operational content of
the dimension--initial-entropy characterization. Known entropy, minimax and
decoupling ingredients do not by themselves supply the complete closed
implementation; combining them under this contract is what the paper adds.
None of the works cited above states a rate region for the closed contract of
Definition~\ref{def:closed-device}.

\section{Conclusion and open questions}
\label{sec:conclusion}

Under the closed contract a single-use dilation question becomes a rate region
in two resources, bath dimension per use and initial entropy per use, with $h$
fixing their sum and $\kappa$ fixing their difference. Three questions are left
open here. Computability: $\kappa$ is an ordered limit of an infimum over
extensions of unbounded dimension, and no algorithm, finite witness bound or
continuity in $\Phi$ is asserted; Proposition~\ref{prop:kappa-lower} gives
only a lower bound that is easy to evaluate, and Proposition~\ref{prop:interior}
leaves a gap of a factor of about three between that bound and the exhibited
upper bound. Efficiency: the all-horizon device of
Theorem~\ref{thm:causal-balancing} is built from a witness that is not
constructed, and no circuit-size or latency claim is made; the streaming
realization for dephasing shows that some corners admit explicit efficient
devices, and which channels do is open. The contract: immediate output return,
one visit per input and a memoryless target are all load-bearing here. Devices
that deliver outputs on a schedule rather than immediately are studied in a
companion note on Schur channels \cite{Douglas2026Schur}; testers with bounded
persistent quantum memory rather than unrestricted references are a separate
question that this paper does not address.

\section*{AI assistance}

AI systems assisted with proof exploration, manuscript preparation and checking.
Additional independent AI sessions audited the arguments and their relationship
to prior work. These sessions were not human external peer review. The scope
of the separate, partial Lean verification is stated in Appendix~\ref{app:diagnostics}.

\appendix

\section{Complete causal balancing proof}
\label{app:balancing}

This appendix proves Theorem~\ref{thm:causal-balancing}. Fix the channel and
one finite \emph{exact} collision before any limit. Use
$\sigma,v,g,h,a,b,R$ from Theorem~\ref{thm:causal-balancing}, with
$0\le a\le\sigma$, $g\ge 0$ and $b=a+g$. All logs are base two and $\Dist$
denotes half trace distance throughout; in Lemma~\ref{lem:seed-decoupling}
the integer dimension of the compressed register is written $D_{\rg{C}}$.

\subsection{Adaptive spectral concentration}

\begin{lemma}
\label{lem:adaptive-projector}
Let $\Lambda$ have finite output and $H_\Lambda=\max_\rho S(\Lambda(\rho))$.
There is a fixed spectral projector on $n$ successively produced, untouched
$\Lambda$ outputs with rank at most $2^{n(H_\Lambda+\gamma)}$ and rejected
weight at most $\exp(-c_\Lambda n\gamma^2)$, uniformly over adaptive
controllers and all finite reference dimensions. Deterministic spectral costs
have zero loss.
\end{lemma}

\begin{proof}
Choose an entropy-maximizing output $\zeta$. Its support contains every output
support: mixing an output with positive weight outside $\supp\zeta$ into
$\zeta$ would increase entropy by a positive leading $-t\log t$ term. On the
common support the directional entropy derivative toward any output $\omega$
is $-\Tr(\omega-\zeta)\log\zeta\le 0$. Thus
\[
\Tr\Lambda(\rho)(-\log\zeta)\ \le\ S(\zeta)=H_\Lambda .
\]
Let $z_j>0$ be its eigenvalues and $L_\Lambda=\max_j(-\log z_j)$.
Hypothetically measure each spent output in this basis. Conditioned on earlier
outcomes, the next input is a density operator and its channel resource is
fresh; therefore $X_i=-\log z_{j_i}$ has conditional mean at most $H_\Lambda$
and lies in $[0,L_\Lambda]$. More explicitly, for a positive-probability
hypothetical history $j_{<i}$, let $\xi_{\rg{S}_i\rg{K}_i\mid j_{<i}}$ be the
normalized conditional state just before visit $i$, with $\rg{K}_i$ containing
the controller and all its references. Freshness makes the current dilation
resource product with this joint state. Writing
$\rho_i(j_{<i})=\Tr_{\rg{K}_i}\xi_{\rg{S}_i\rg{K}_i\mid j_{<i}}$, one has
\[
\Pr(j_i=j\mid j_{<i})=\bra{j}\Lambda(\rho_i(j_{<i}))\ket{j},
\qquad
\mathbb{E}[X_i\mid j_{<i}]=\Tr\Lambda(\rho_i(j_{<i}))(-\log\zeta)\le H_\Lambda.
\]
The conditional input may depend on the full history and be mixed because of
coherent entanglement; neither fact changes the channel-output formula.
Zero-probability histories may be omitted. These measurements are proof
devices, not information supplied to the controller. The conditional
exponential-moment bound follows directly by iterating
$\mathbb{E}[\exp(t(X_i-\mathbb{E}[X_i\mid\text{past}]))\mid\text{past}]
\le\exp(t^2L_\Lambda^2/8)$, then minimizing
$\exp(-tn\gamma+nt^2L_\Lambda^2/8)$ over $t\ge 0$. It gives
\[
\Pr\Bigl\{\sum_i X_i>n(H_\Lambda+\gamma)\Bigr\}
\le\exp(-2n\gamma^2/L_\Lambda^2).
\]
Indeed $\mathbb{E}[\exp(t(X_i-H_\Lambda))\mid\text{past}]
\le\exp(t^2L_\Lambda^2/8)$ for $t\ge 0$, which gives the iteration even when
the conditional means vary. Thus $c_\Lambda=2/L_\Lambda^2$ is independent of
the tester; it need not be uniform over different channels or witnesses with
small positive eigenvalues.

This is conditional iteration of the bounded-variable exponential-moment
estimate underlying Hoeffding's inequality \cite{Hoeffding1963}. For
completeness, the log moment-generating function of a variable in $[0,L]$ has
second derivative equal to its variance in an exponentially tilted law, at
most $L^2/4$. Integrating twice from zero gives the centered bound $t^2L^2/8$,
also for each conditional law. The application here selects one spectral cost
from the entropy-maximizing output, uniformly over controllers.

For $L_\Lambda=0$ the event is empty. The hypothetical measurements commute
with subsequent service and controller operations up to the boundary just
before compression, because these operations do not touch spent outputs.
Thus the same probability is the rejected weight of the unmeasured state
at that boundary. No commutation with the later compressor or encoder is
asserted or needed. Every
accepted eigenstring has product $\zeta$ weight at least
$2^{-n(H_\Lambda+\gamma)}$, so there are at most $2^{n(H_\Lambda+\gamma)}$
such strings. This proves the claim without i.i.d.\ inputs or a union bound
over users.
\end{proof}

For $\Lambda=\Gamma_C$, fresh $\tau$ cells therefore give a user-independent
$P_{\rg{B}}$ on the spent $\rg{B}^n$ cells with rank at most
$2^{n(v+\gamma)}$ and $\mathrm{loss}_{\rg{B}}\le\exp(-c_{\rg{B}}n\gamma^2)$.
All physical basis changes and compressors depend only on this fixed witness,
$n$ and $\gamma$. The same lemma applies to untouched minimal environments of
$\Phi$.

For comparison, Metger et al.\ \cite[Theorem 4.1 and Appendix A]{Metger2024}
bound sequential
smooth entropies. Their min-entropy theorem requires
$\Tr_{\rg{A}_i\rg{R}_i}M_i=\mathcal{R}_i'\Tr_{\rg{R}_{i-1}}$; Appendix A
Eq.~(A.2) there gives a max-entropy form without that restriction. Such
statewise bounds alone do not specify one compression projector shared by
every tester. Here the fixed $P_{\rg{B}}$ preserves all complete adaptive
tests in Definition~\ref{def:complete-error-rates} through the trace-distance
argument; Lemma~\ref{lem:exterior-cuts} additionally constructs simultaneous
support and spectral-cap cuts. We use the direct proof above, not an
entropy-accumulation citation, to infer projector uniformity or seed
freshness.

\subsection{The complete purified exterior}

\begin{lemma}
\label{lem:exterior-cuts}
In an exact comparison block there are commuting proof projectors on
$\rg{B}^n,\rg{H},\rg{F}^n$ whose success $p$ is at least
$1-\mathrm{loss}_{\rg{B}}-\mathrm{loss}_{\rg{E}}-\mathrm{loss}_{\rg{F}}$, and
whose normalized pure projection obeys
\[
\rank\hat\rho_{\rg{H}}\le K=2^{n(h+\gamma)},
\qquad
\lVert\hat\rho_{\rg{B}^n\rg{H}}\rVert_\infty\le 2^{-n(\sigma-\gamma)}/p .
\]
Here $\rg{H}$ is a purification of the \emph{complete} future-accessed
exterior, excluding the independent seed and current source. No rank bound on
its raw mixed marginal is asserted.
\end{lemma}

\begin{proof}
At a comparison boundary let $\rg{E}_{\mathrm{raw}}$ contain the entire user
and references, all unused stock, and every other spectator that will be
accessed again. Exclude the seed $\rg{Z}$, the $n$ current source cells, and
permanently parked registers. The source $\tau^{\otimes n}$ is independent of
$\rg{E}_{\mathrm{raw}}$ and $\rg{Z}$ jointly. Purify
$\rg{E}_{\mathrm{raw}}$ to $\rg{H}=\rg{E}_{\mathrm{raw}}\rg{E}'$ for analysis,
and purify each source with $\rg{F}_i$. Dilate the tester's operations,
retaining all its records in $\rg{H}$. The final
$\rg{B}^n\rg{F}^n\rg{H}$ state is pure, and the $\rg{F}^n$ marginal stays
$\tau_{\rg{F}}^{\otimes n}$. The device acts on neither $\rg{F}^n$ nor
$\rg{E}'$. These are comparison purifications, not extra physical workspace or
modifications of an actual inert purifier.

Run the same purified controller with fresh pure minimal dilations of $\Phi$
instead. The resulting virtual state $\theta_{\rg{E}^n\rg{H}}$ is pure and its
$\rg{H}$ marginal equals the actual comparison $\rg{H}$ marginal: both provide
exactly $\Phi$ at every visit, including to reference-entangled inputs.
Lemma~\ref{lem:adaptive-projector} gives a virtual $P_{\rg{E}}$ of rank at
most $K$ and $\mathrm{loss}_{\rg{E}}\le\exp(-c_{\rg{E}}n\gamma^2)$. Define
\[
X_{\rg{H}}=\Tr_{\rg{E}^n}\bigl[(P_{\rg{E}}\otimes\one)
\ketbra{\theta}{\theta}(P_{\rg{E}}\otimes\one)\bigr].
\]
Then $0\le X_{\rg{H}}\le\rho_{\rg{H}}$, since partial trace kills the cross
terms between $P_{\rg{E}}$ and $\one-P_{\rg{E}}$, and
$\rank X_{\rg{H}}\le\rank P_{\rg{E}}\le K$. If
$P_{\rg{H}}=\supp X_{\rg{H}}$, then
$\Tr P_{\rg{H}}\rho_{\rg{H}}\ge\Tr X_{\rg{H}}\ge 1-\mathrm{loss}_{\rg{E}}$.
The projector $P_{\rg{H}}$ may depend on the user and is only a proof
projector. An initially mixed spectator may have arbitrarily large raw rank
even when $h=0$; including $\rg{E}'$ before this argument is essential.

The positive eigenvalues of $\tau$ have finite logarithmic costs. The ordinary
lower-tail bound for independent source eigenstrings supplies $P_{\rg{F}}$
accepting $-\log p_x\ge n(\sigma-\gamma)$, with
$\mathrm{loss}_{\rg{F}}\le\exp(-c_{\rg{F}}n\gamma^2)$, and
\[
P_{\rg{F}}\,\tau_{\rg{F}}^{\otimes n}\,P_{\rg{F}}
\le 2^{-n(\sigma-\gamma)}P_{\rg{F}} .
\]
Zero eigenvalues have zero probability and are omitted. Deterministic-cost
tails have zero loss. The projectors $P_{\rg{B}},P_{\rg{H}},P_{\rg{F}}$ act on
distinct registers and commute; the union bound gives the stated $p$.
Crucially the virtual $P_{\rg{E}}$ is not applied to the mixed-collision
purification alongside a physical bath projector.

For every vector on $\rg{F}$, the quadratic form of the $\rg{F}$ marginal
after projecting $\rg{B}\rg{H}$ is at most its original value, because
$P_{\rg{B}}P_{\rg{H}}$ is a positive contraction. Thus that unnormalized
$\rg{F}$ marginal is at most $\tau_{\rg{F}}^{\otimes n}$. After $P_{\rg{F}}$
and normalization it is capped by $2^{-n(\sigma-\gamma)}/p$. Purity transfers
this nonzero spectrum to $\rg{B}^n\rg{H}$. The normalized projected state has
half distance $\sqrt{1-p}$ from the original pure state. No independence of
spectral events is needed, and the factor $1/p$ cannot be suppressed.
\end{proof}

\subsection{Second moment with a retained physical seed}

\begin{lemma}
\label{lem:seed-decoupling}
Let $\rg{C}=\rg{M}\rg{G}$, $\dim\rg{C}=D_{\rg{C}}$, $\dim\rg{M}=m$, and
$\rho_{\rg{C}\rg{H}}$ be normalized with $\rank\rho_{\rg{H}}\le K$ and
$\lVert\rho_{\rg{C}\rg{H}}\rVert_\infty\le\lambda$. A uniform exact unitary
two-design $V_z$ on $\rg{C}$ gives
\[
\mathbb{E}_z\bigl\lVert\Tr_{\rg{G}}(V_z\rho_{\rg{C}\rg{H}}V_z^*)
-\pi_{\rg{M}}\otimes\rho_{\rg{H}}\bigr\rVert_1
\le\sqrt{m^2K\lambda/D_{\rg{C}}}.
\]
If an independent actual classical seed $\rg{Z}$ is retained, the half trace
distance of $\rg{Z}\rg{M}\rg{H}$ from
$\pi_{\rg{Z}}\otimes\pi_{\rg{M}}\otimes\rho_{\rg{H}}$ is at most half this
expression.
\end{lemma}

\begin{proof}
For $\Delta_V=\Tr_{\rg{G}}(V\rho_{\rg{C}\rg{H}}V^*)
-\pi_{\rg{M}}\otimes\rho_{\rg{H}}$, the two-copy Haar twirl of
$\mathrm{Swap}_{\rg{M}}\otimes\one_{\rg{G}\rg{G}}$ equals
$\alpha\one+\beta\,\mathrm{Swap}_{\rg{C}}$. The untwirled operator has trace
$D_{\rg{C}}^2/m$ and trace against $\mathrm{Swap}_{\rg{C}}$ equal to
$D_{\rg{C}}m$; hence, for $D_{\rg{C}}>1$,
\[
\beta=\frac{D_{\rg{C}}(m^2-1)}{m(D_{\rg{C}}^2-1)}\le\frac{m}{D_{\rg{C}}},
\qquad \alpha=\frac{1}{m}-\frac{\beta}{D_{\rg{C}}}.
\]
The swap trick and the cross term against
$\pi_{\rg{M}}\otimes\rho_{\rg{H}}$ give exactly
\[
\mathbb{E}\Tr\Delta_V^2=\beta\Bigl[\Tr\rho_{\rg{C}\rg{H}}^2
-\Tr\rho_{\rg{H}}^2/D_{\rg{C}}\Bigr].
\]
$\Delta_V$ is supported on $\rg{M}\otimes\supp\rho_{\rg{H}}$, of dimension at
most $mK$. Cauchy--Schwarz for its singular values, Jensen, and
$\Tr\rho_{\rg{C}\rg{H}}^2\le\lambda$ give the claimed trace-norm bound. If
$m=1$, then $\Delta_V=0$ exactly, including $D_{\rg{C}}=1$. An exact
two-design has these same second moments. This is the elementary
specialization underlying Dupuis et al.\
\cite[Theorem 3.3 and Lemmas 3.4--3.5]{Dupuis2014}; the factor $1/2$ is
explicit here.

For $N$ qubits, the uniform Clifford group modulo global phases is such a
finite exact two-design, by Dankert et al.\
\cite[Theorem 1, pp.~2--3]{Dankert2009}. Its size
$J_N$ has $\log J_N\le 2N(2N+1)$: each of $2N$ Pauli generators has at most
$2^{2N+1}$ signed images, and these images determine conjugation. For $N=0$
use $J_0=1$. Initialize $\rg{Z}$ in $\pi_{J_N}$ offline and apply
$\sum_z\ketbra{z}{z}\otimes V_z$. If $\rg{Z}$ is independent of
$\rg{C}\rg{H}$ \emph{jointly} beforehand, the cq block trace-norm identity
gives
\[
\Dist(\rho_{\rg{Z}\rg{M}\rg{H}},
\pi_{\rg{Z}}\otimes\pi_{\rg{M}}\otimes\rho_{\rg{H}})
=\tfrac12\,\mathbb{E}_z\lVert\Delta_{V_z}\rVert_1 .
\]
Keeping the random unitary choice in the output is also explicit in Dupuis et al.\
\cite[discussion after Corollary 3.2]{Dupuis2014}. The additional accounting
here charges this physical label and reuses it across blocks through the joint
comparison in Proposition~\ref{prop:joint-recycling}.

No one seed value is selected for all users. This is a guarantee for each user
on a physically retained mixed seed, with constants uniform over users. The
seed's inert purifier is not included in this reduced independence claim; it
is never accessed. Marginal independence of $\rg{Z}$ from $\rg{C}$ and from
$\rg{H}$ separately would not suffice.
\end{proof}

\subsection{Exact integer wires and reversible preparation}
\label{app:integer-wires}

\begin{lemma}
\label{lem:integer-wires}
Fix $n\ge 1$ and $\gamma>0$. Define integer widths
\[
\begin{gathered}
N=\lceil n(v+\gamma)\rceil,\quad
m_i=\lceil n(\sigma+2\gamma)\rceil,\quad
m_o=\max(0,\lfloor n(\sigma-a-6\gamma)\rfloor),\quad
r_p=\lceil 3n\gamma\rceil+4,\\
c=m_i-m_o,\quad d=N+r_p-m_i,\quad z=c+d=N-m_o+r_p,\\
e=n\lceil\log R\rceil,\quad w=\max(0,e-N).
\end{gathered}
\]
The full block can be implemented on permanent $\rg{M},\rg{W}$ of $m_o,w$
qubits and one disjoint $z$-qubit tranche. The tranche starts with $c$ mixed
and $d$ pure qubits. Both preparation and encoding residues end in that
tranche, which is never accessed again. All dimensions in this statement are
physical.
\end{lemma}

\begin{proof}
The typical-spectrum threshold is the usual source-coding ingredient; see Abeyesinghe et al.\
\cite[section VII and Appendix A]{Abeyesinghe2009}. Here the preparation is
implemented as an injection of flat labels with all its fibres retained on the
counted wires below. Zero-pad the $R$-dimensional cell to $\lceil\log R\rceil$
qubits and extend $U$ unitarily outside its valid sector. From
$0\le a\le\sigma$, $g=v-\sigma\ge 0$ and $v\le\log R$, we have $c,d\ge 0$,
$m_o\le m_i$, $m_o\le N$, $m_i\le e+r_p$ and $z\ge r_p$. In particular
$d\ge ng+2n\gamma+3$ and $m_i\le e+2n\gamma+1\le e+r_p$. Clipping $m_o$
changes its ideal leading term $n(\sigma-a)$ by at most $6n\gamma+1$.
Consequently, uniformly over clipping,
\[
c=na+O(n\gamma+1),\qquad d=ng+O(n\gamma+1),\qquad z=nb+O(n\gamma+1).
\]
At opening, $\rg{M}$ and the fresh $c$ mixed tranche qubits give exactly $m_i$
mixed qubits, while all other local wires are pure. Their exact volume
identity is
\begin{equation}
\label{eq:wire-volume}
m_i+d+w=N+r_p+w=e+r_p+\max(0,N-e)=m_o+z+w .
\end{equation}
In the eigenbasis of $\tau^{\otimes n}$, for strings with probabilities in
$[2^{-n(\sigma+\gamma)},2^{-n(\sigma-\gamma)}]$, assign
$\lfloor 2^{m_i}p_x\rfloor$ distinct flat input labels to cell label $x$ and
different preparation-residue labels. Each fibre uses at most
$2^{3n\gamma+1}<2^{r_p}$ labels. There are at most $2^{n(\sigma+\gamma)}$
typical strings, so their total rounding deficit is at most
$2^{n(\sigma+\gamma)-m_i}\le 2^{-n\gamma}$. The desired subdistribution is
pointwise below the target distribution. Send leftover labels injectively to
unused cell/residue pairs; $m_i\le e+r_p$ guarantees capacity. Although these
labels can raise some marginal probabilities above target, total variation is
at most the missing subdistribution mass, giving
\[
\delta_{\mathrm{prep}}\le\Pr\{\text{source atypical}\}+2^{-n\gamma}.
\]
The two-sided source tail is at most $2\exp(-c_\tau n\gamma^2)$, with
deterministic costs treated exactly. The injection of an orthonormal input set
extends to a unitary on the full local space in \eqref{eq:wire-volume}; any
$\max(0,N-e)$ padding stays pure. The complete local positive spectrum remains
exactly $2^{m_i}$ copies of $2^{-m_i}$. Only its cell marginal is
approximately nonflat.

Route the $r_p$ preparation-residue positions into the current tranche and
never touch them again. In the comparison run replace only the cell marginal
by $\tau^{\otimes n}$; all future-accessed exterior registers remain
independent. The preparation error is charged once for the entire block by
contractivity, even with adaptive service. This is a comparison of reduced
states, not a physical erasure or operation on the purifying systems.

Run each exact collision when its input arrives. At an interior visit,
hand back its output after that collision. At the final visit of a block,
perform the collision, then the bath-only closing gates described below,
and then release that visit's output, all before the next input arrives.
Bath-only unitaries preserve the joint reduced state of every user output
and reference, including the current output. On $e$ cell positions and
$\max(0,N-e)$ padding, a fixed
unitary compresses $\supp P_{\rg{B}}$ into $N$ qubits $\rg{C}$ and $w$ pure
qubits $\rg{W}$. The rank bound in Lemma~\ref{lem:adaptive-projector}
guarantees this embedding; complete its remaining columns orthonormally. No
projection is performed or success flag measured. Split $\rg{C}$ into $\rg{M}$
of $m_o$ qubits and $\rg{G}$ of $N-m_o$ qubits, use
Lemma~\ref{lem:seed-decoupling}'s controlled encoder, and route $\rg{M},\rg{W}$
to their permanent positions. Park $\rg{G}$ beside the preparation residue in
the same tranche. The equality
\begin{equation}
\label{eq:tranche-split}
z=r_p+(N-m_o)
\end{equation}
proves that neither residue overlaps the other or any future stock. The
compression failure part remains coherently present on these same wires.
\end{proof}

\subsection{Joint recycling, every prefix and the repeated unitary}

\begin{proposition}
\label{prop:joint-recycling}
The block above has a user-independent reduced-state error
$\epsilon_n\le A[\exp(-c_0n\gamma^2)+2^{-c_1n\gamma}]$ for large $n\gamma^2$,
returning $\rg{M},\rg{W}$ and the seed jointly independent of the entire
future-accessed exterior in the comparison state. Constants depend only on the
fixed collision.
\end{proposition}

\begin{proof}
Apply Lemma~\ref{lem:exterior-cuts} to the exact comparison block. After
compression its projected $\rg{W}$ is pure and its normalized
$\rg{C}\rg{H}$ state has rank bound $K$ and cap
$\lambda=2^{-n(\sigma-\gamma)}/p$. Before encoding the seed is independent of
$\rg{C}\rg{H}$ jointly: preparation and service have not touched it, and the
proof cuts can be chosen without referring to it. In particular the comparison
input is $\pi_{\rg{Z}}\otimes\rho_{\mathrm{source},\mathrm{exterior}}$, with
source and exterior also in product. Its purification and $P_{\rg{H}}$ may
depend on the tester but not on $\rg{Z}$; $P_{\rg{B}}$ depends only on the
fixed collision. Applying these proof cuts therefore preserves the product
with $\pi_{\rg{Z}}$. This is the joint hypothesis of
Lemma~\ref{lem:seed-decoupling}, including all unused stock and
future-accessed spectators in the exterior. For $m_o>0$,
\[
2m_o+n(h+\gamma)-n(\sigma-\gamma)-N\le-11n\gamma,
\]
since $\sigma-2a+h-v=0$. Lemma~\ref{lem:seed-decoupling} therefore bounds the
projected joint half distance by $2^{-11n\gamma/2}/(2\sqrt{p})$. For $m_o=0$
it is zero. Add $\sqrt{1-p}$ for smoothing the initial pure comparison, and
another for replacing its projected $\rg{H}$ marginal by its original $\rg{H}$
marginal. $\rg{W}$ is pure on the projected branch, so it is included jointly
in the same estimate. Tracing the mathematical $\rg{E}'$ then gives the bound
on the physical exterior. Including preparation once yields
\begin{equation}
\label{eq:block-error}
\epsilon_n\le\delta_{\mathrm{prep}}
+2\sqrt{\mathrm{loss}_{\rg{B}}+\mathrm{loss}_{\rg{E}}
+\mathrm{loss}_{\rg{F}}}+\frac{2^{-11n\gamma/2}}{2\sqrt{p}}
\end{equation}
when $m_o>0$; omit the last term otherwise. With $p\ge 1/2$ the asserted bound
follows, absorbing square roots into fixed positive constants.

For $L$ full blocks the boundary comparison is
\begin{equation}
\label{eq:boundary-state}
\rho_{\mathrm{user},\mathrm{ideal}}\otimes\pi_{\rg{Z}}\otimes\pi_{\rg{M}}
\otimes\ketbra{0}{0}_{\rg{W}}\otimes(\text{all unused mixed/pure tranches})
\otimes(\text{pure tail cells}).
\end{equation}
Old tranches are omitted only from this trace-distance metric; their full
dimensions remain allocated. Initially the actual reduced state is
\eqref{eq:boundary-state}, so $e_0=0$. Fix any adaptive tester, apply the
\emph{same} next-block CPTP map to its actual state and to
\eqref{eq:boundary-state}, and contract their distance. The uniform local
bound then gives $e_{j+1}\le e_j+\epsilon_n$. Thus actual learning of the seed
and correlations with recycled $\rg{M}$ are charged to $e_j$. Freshness is
used only in the comparison run, where the next source and seed are jointly
independent of the complete exterior. There is no assumption about
independence of errors.

At a prefix inside a block before its closing encoder, the user error is at
most $e_j+\delta_{\mathrm{prep}}$. The closing bath-only gates can be
performed after the last collision and before handing back that visit's
output, preserving immediate service. Fewer than $n$ tail visits use
preallocated pure minimal-Stinespring cells, exactly implementing $\Phi$;
these are part of the initial bath. Thus the entire horizon and every prefix
obey the conservative bound
$\delta_T\le L\epsilon_n+\delta_{\mathrm{prep}}$. When the final segment
consists only of pure tail cells, no extra preparation term is needed.
Stopping is handled by extending a stopped tester with dummy inputs while
keeping its records and stop flag, then tracing the continuation. This is
unconditional; no rare-event postselection estimate follows.

All opening, collision, closing and routing steps for visit $t$ form a fixed
unitary $V_t$ on the system and allocated non-clock bath. A pure counted
$(T+1)$-dimensional clock implements the \emph{same} unitary each visit:
\begin{equation}
\label{eq:repeated-unitary}
W_T=\sum_{t=0}^{T}\ketbra{t+1\bmod (T+1)}{t}_{\mathrm{clock}}\otimes V_t .
\end{equation}
Choose the unused $V_T$ arbitrarily unitary. Orthogonality of clock sectors
and unitarity of each $V_t$ prove $W_T$ is unitary. Starting at $0$, the first
$T$ visits do not wrap. All control resides in this fixed hardware. There is
no bound on gate count or latency, but the next input never arrives before the
previous output is handed back.
\end{proof}

\subsection{Actual spectrum and the every-horizon limit}

\begin{proposition}
\label{prop:actual-spectrum}
The construction has the exact initialization and separate limits asserted in
Theorem~\ref{thm:causal-balancing}.
\end{proposition}

\begin{proof}
Set $L=\lfloor T/n\rfloor$ and $t_{\mathrm{tail}}=(T-Ln)\lceil\log k\rceil$,
where $k=\rank J_\Phi$ is the minimal environment dimension. Zero-padding each
pure tail cell and unitary completion are counted. The exact numbers of
allocated data qubits and initially mixed data qubits are
\begin{equation}
\label{eq:data-widths}
A_T=m_o+w+Lz+t_{\mathrm{tail}},\qquad K_T=Lc+m_o .
\end{equation}
Since $d\ge 0$, $K_T\le A_T$. With the clock and the one reused seed,
\begin{equation}
\label{eq:initial-spectrum}
\begin{gathered}
R_T=(T+1)J_N2^{A_T},\qquad
\omega_T=\ketbra{0}{0}_{\mathrm{clock}}\otimes\pi_{J_N}\otimes
\pi_{2^{K_T}}\otimes\ketbra{0}{0}_{\mathrm{rest}},\\
P_T=J_N2^{K_T},\qquad S(\omega_T)=\log J_N+K_T .
\end{gathered}
\end{equation}
Its positive eigenvalues are exactly $1/P_T$ repeated $P_T$ times, with
$R_T-P_T$ zeros. Its minimal inert purifier has dimension $P_T$ and is never
accessed. The witness entropy $\sigma$ is an intermediate preparation
parameter; it is not substituted for the actual entropy in
\eqref{eq:initial-spectrum}.

For every sufficiently large integer $T$ take $n=\lfloor T^{1/3}\rfloor$ and
$\gamma=n^{-1/4}$. Fixed-witness constants have already been chosen. Then
$n\gamma^2=\sqrt{n}$, so the exponential in \eqref{eq:block-error} beats the
polynomial number $L$ of blocks; the other exponent is of order $n^{3/4}$.
Thus
\[
L\epsilon_n+\delta_{\mathrm{prep}}\to 0,\quad
A_T/T\to b,\quad K_T/T\to a,\quad
\log J_N/T=O_C(n^2/T)\to 0,
\]
\[
(m_o+w+t_{\mathrm{tail}})/T=O_C(n/T)\to 0,\qquad \log(T+1)/T\to 0 .
\]
Slack costs $O(T\gamma)$, rounding costs $O(T/n)$, and the seed $O_C(n^2)$
bits; all are sublinear. Finitely many smaller horizons use preallocated exact
pure cells and a counted clock. Consequently for every fixed witness and every
tolerance there is a finite threshold after which \emph{every} horizon has
small error and small rate discrepancies, including the separate seed and
mixed-data limits. No effective or witness-uniform threshold is asserted.

The same integers handle $m_o=0$, $g=0$, $a=0$, $R=1$ and $N>e$. For example
pure $\tau$ forces $\sigma=a=0$; its mixed slack consumption is $o(T)$. An
identity collision with an arbitrary mixed spectator has $g=h=a=b=0$ and its
startup stock is $O_C(n)$. These checks use no negative register widths.
\end{proof}

The argument uses encoder-only fully quantum Slepian--Wolf, with its would-be
message parked and counted; see Section~\ref{sec:balancing} for attribution.
All original proof steps needed beyond that ingredient have been supplied
above. This concludes the proof of Theorem~\ref{thm:causal-balancing}.

\section{Conventions, resources and references}
\label{app:conventions}

This appendix records conventions, registers, reproducible finite checks and
the bibliography. It adds no claims.

\subsection{Normalization and finite objects}

\begin{center}
\small
\begin{tabular}{@{}>{\raggedright\arraybackslash}p{3.0cm}>{\raggedright\arraybackslash}p{11.4cm}@{}}
\hline
\textbf{Symbol} & \textbf{Convention}\\
\hline
$\rg{S},\rg{A},\rg{Q}$ & Input, immediate output and purification reference;
dimension $q$ each.\\
$\rho_{\rg{Q}}$ & $\rho^{T}$ for the canonical purification
$(\one\otimes\sqrt{q\rho})\Omega_q$; spectra and minimum eigenvalues agree
with $\rho$.\\
$J_\Phi$ & Trace-one Choi state, $\Tr_{\rg{A}}J_\Phi=\one_{\rg{Q}}/q$;
$k=\rank J_\Phi\le q^2$.\\
$\Dist,\Ddiam$ & Half trace norm and half diamond norm; range $[0,1]$ for
states and channels.\\
$S,H_2$ & Base-two von Neumann and binary entropy; $0\log 0=0$.\\
$\pi_d$ & $\one_d/d$, with $d$ an actual integer Hilbert dimension.\\
$k_{\rho,\varepsilon}$ & All normalized finite-$\rg{Q}\rg{A}\rg{Z}$
extensions with exact $\rg{Q}\rg{Z}$ product and $\rg{Q}\rg{A}$ error at most
$\varepsilon$; auxiliary dimension unrestricted.\\
$\kappa$ & Full-rank input supremum after the positive-error limit at fixed
$\rho$; Corollary~\ref{cor:pointwise-funnel} derives the exact privacy-funnel
form using Theorems~\ref{thm:uniform-gain} and \ref{thm:active-repair},
without attained witnesses.\\
$K_0(\rho),\Pfun$ & Exact fixed-input extension infimum and quantum
privacy-funnel supremum with $X=\rg{E},Y=\rg{Q},R=\rg{A},W=\rg{Z}$;
unrestricted finite auxiliary dimension. The prior subscript $\mathrm{q}$
means quantum.\\
$\Gamma_C$ & Active bath output only; no inert purifier included.\\
$r,s$ & Actual limits $\log R_T/T$ and $S(\omega_T)/T$ for vanishing-error
all-horizon families.\\
Entropy deficit & $\log R_T-S(\omega_T)$; neither entropy nor log rank.\\
Appendix~\ref{app:balancing} constants & May depend on a single fixed finite
exact collision; uniform in tester, finite reference dimension and horizon.\\
Reused letters & $\rg{F}$ is the inert purifier of the actual initializer
(Definition~\ref{def:closed-device}) and, as $\rg{F}^n$ in
Appendix~\ref{app:balancing}, the purifier of the comparison stock; the
Slepian--Wolf receiver is written $\rg{N}$. $C$ names both a finite collision
$(U,\tau)$ and the Slepian--Wolf sender register $\rg{C}$. $k$ is the Choi
rank; $k_{\rho,\varepsilon}$ is the smoothed extension cost. $\nu$ in
Proposition~\ref{prop:interior} is a damping parameter; $\gamma$ in
Appendix~\ref{app:balancing} is the concentration slack.\\
\hline
\end{tabular}
\end{center}

The case $q=1$ is scalar and is directly realized without a bath. In
Theorem~\ref{thm:active-repair} one uses the smallest \emph{positive} target
Choi eigenvalue, never a kernel inverse. Square input and output dimensions
are needed for the support-isometry unitary completion in
Theorems~\ref{thm:uniform-gain} and \ref{thm:active-repair}. All finite
internal direct sums and zero padding are included in the active dimension.

\subsection{Physical and proof registers of Appendix~\ref{app:balancing}}

\begin{center}
\small
\begin{tabular}{@{}>{\raggedright\arraybackslash}p{2.6cm}>{\raggedright\arraybackslash}p{4.0cm}>{\raggedright\arraybackslash}p{7.8cm}@{}}
\hline
\textbf{Register} & \textbf{Size / initialization} & \textbf{Action and charge}\\
\hline
Full $\rg{B}_T$ & $R_T=(T+1)J_N2^{A_T}$ & Every active device register below
is included.\\
Clock & $T+1$; pure & Cyclic shift controls the fixed per-visit unitaries;
$\log(T+1)$ charged.\\
$\rg{Z}$ & $J_N$; $\pi_{J_N}$ & One retained seed reused for every block;
$\log J_N$ charged in dimension and entropy.\\
$\rg{M}$ & $m_o$ qubits; maximally mixed & Permanent recycled capital,
returned jointly in the reduced comparison.\\
$\rg{W}$ & $w$ qubits; pure & Permanent workspace; returned pure jointly up to
the block error.\\
Tranche $j$ & $z=c+d$ qubits; $c$ mixed and $d$ pure & Distinct for every full
block. Preparation residue $r_p$ and encoding residue $N-m_o$ fill it
permanently.\\
Working cells & $e=n\lceil\log R\rceil$ positions during a block & Obtained by
routing within the exact local wires, not a separate free bank.\\
Compressor padding & $\max(0,N-e)$ pure positions & Included in identity
\eqref{eq:wire-volume}; handles $N>e$.\\
Tail & $(T-Ln)\lceil\log k\rceil$ pure qubits & Preallocated exact
pure-Stinespring cells; retained after use.\\
Actual $\rg{F}$ & Minimal dimension $J_N2^{K_T}$ & Purifies the complete
initializer and is forever inert; not an active resource.\\
Proof $\rg{F}^n$ & Purifies comparison $\tau^{\otimes n}$ & Not the actual
repeated stock initializer and never an operated register.\\
$\rg{H}=\rg{E}_{\mathrm{raw}}\rg{E}'$ & Purification of the complete
comparison exterior & Includes all future-accessed spectators before its rank
bound; $\rg{E}'$ is mathematical only.\\
Virtual $\rg{E}^n$ & Fresh minimal environments in a comparison & Produces
support projector $P_{\rg{H}}$; virtual $P_{\rg{E}}$ is never a physical
compression operation.\\
\hline
\end{tabular}
\end{center}

Only the reduced boundary distance omits parked tranches and inaccessible
purifiers. Their omission from a metric does not delete any actual bath
dimensions. Exact initialization is $P_T$ equal eigenvalues $1/P_T$, where
$P_T=J_N2^{K_T}\le R_T$; its remaining $R_T-P_T$ eigenvalues are zero.

\subsection{Scope of the statements}

The physical contract for every statement in this paper is
Definitions~\ref{def:closed-device}--\ref{def:collision}. Nothing above uses a
theorem about finite-tracial closure, about the strength of a physical record,
about feedback-assisted protocols, about computability, or about processes
with more than one visit per input. Propositions~\ref{prop:unitary}--%
\ref{prop:interior} and \ref{prop:shared-seed} are proved here under those
same definitions. No external classification of generalized amplitude damping,
arbitrary mixed replacers or unequal pure-output sectors is used.

Four shortcuts are deliberately avoided and are used nowhere above: raw
mixed-history rank in place of purified-exterior support, merely marginal seed
freshness, uniform initial-entropy continuity under a small flag, and attained
witnesses.

\subsection{Reproducible diagnostics and their limits}
\label{app:diagnostics}

The scripts named below accompany this paper in \texttt{ancillary/}; run the
Python commands from that directory. For Lean, use the wrapper command in
\texttt{ancillary/COVERAGE.md} with an already installed compatible package
cache; the table names its target files, not standalone commands with configured
imports. They print their results; all are bounded
finite controls. They do not prove unbounded auxiliary optimizations,
arbitrary adaptive asymptotics, or any theorem of this paper. Seeds and
dimensions are explicit in the scripts.

\begin{center}
\small
\begin{tabular}{@{}>{\raggedright\arraybackslash}p{5.6cm}>{\raggedright\arraybackslash}p{8.8cm}@{}}
\hline
\textbf{Command in the ancillary bundle} & \textbf{Finite coverage}\\
\hline
\texttt{python c2\_lifting\_mixing.py} & 24 lifts, 20 flagged probes;
$q=2,3$, $R=2,3,4$; seed 20260908; tolerance $10^{-8}$.\\
\texttt{python c3\_exactification.py} & 8 cases, 40 probes; $q=2,3$, original
$R=5,7$; largest unitary 93; seed 20260909; tolerance $10^{-8}$.\\
\texttt{python c4\_adaptive\_spectral.py} & 8 nonproduct distributions, 112
conditional nodes; $q=R=2$, $n=3$; vector dimension $\le 1024$; seed 20260909;
tolerance $10^{-8}$.\\
\texttt{python c5\_region\_spectrum.py} & 12 spectrum/feedback cases, 576
rational cases; seed 20260909; tolerance $10^{-9}$; density dimension
$\le 768$.\\
\texttt{python support\_repair.py} & 12 small-matrix cases; seed 20260908;
tolerance $10^{-8}$.\\
\texttt{python block\_invariant.py} & 300 integer allocations, 1024-label
preparation, joint-freshness and rare-stop negative controls, 12 states over
24 Cliffords.\\
\texttt{python encoder\_only\_merging.py} & 36 partial-trace decoder-omission
identity checks and 125 rational balances; seed 20260908; no decoupling or
compiler validation.\\
\texttt{python rz\_contract\_negative\_control.py} & Exact rational $q=R=2$,
$T=1,\dots,12$ marginal/joint separation; no random seed or numerical
tolerance.\\
\texttt{python gad\_interior\_example.py} &
Proposition~\ref{prop:interior}: trace preservation, exchange-block
unitarity, 2007 inputs reproducing the channel, both analytic endpoints, the
exact rational lower witness, the attained upper witness, the mixed-input
entropy-exchange value and diagonal grid searches for both; seed 20260910;
tolerance $10^{-8}$.\\
\texttt{BlockAccounting.lean} (via documented wrapper) & Six elementary natural-number accounting
lemmas only.\\
\texttt{ActiveBath.lean} (via documented wrapper and pinned mathlib) & Finite complex density
matrices and exact bath-only closing invariance for the complete user marginal.\\
\hline
\end{tabular}
\end{center}

The Lean files report their axioms and use only standard foundations
(\texttt{propext}, \texttt{Quot.sound} and \texttt{Classical.choice}), with no
\texttt{sorry} or added scientific axiom. ActiveBath defines positive-semidefinite
trace-one complex matrices, proves that partial trace and unitary closing
preserve legal states, and proves exact invariance of the complete user marginal
under a bath-only unitary, including arbitrary user--bath correlations and finite
references. This supports only the closing identity in Appendix~\ref{app:integer-wires},
not construction of the compiler. Coverage is \emph{partial}: entropy, trace distance
and the asymptotic theorem are not formalized. No numerical telescope suite for
Theorem~\ref{thm:same-spectrum-converse} is claimed. The ancillary coverage matrix
and pinned dependency record give exact reproduction details.

\subsection{Primary bibliography and exact uses}
\label{app:bibliography}
\begingroup
\renewcommand{\section}[2]{}

\endgroup

The bibliography is a bounded ingredient and contract comparison, not an
exhaustive literature survey. Imported primary theorems are cited at their
specific uses. The original arguments are in the main text and
Appendix~\ref{app:balancing}.

\end{document}